\documentclass[12pt]{article}

\usepackage[a4paper, tmargin=1.1in, lmargin=0.7in, rmargin=0.7in, bmargin=1.1in]{geometry}

\usepackage{inputenc}
\usepackage{amsmath}
\usepackage{graphicx,psfrag,epsf, caption}
\usepackage{enumerate}
\usepackage{amsmath}
\usepackage{amssymb}
\usepackage{bm}
\usepackage{hyperref}
\usepackage[T1]{fontenc}
\usepackage{enumerate}
\usepackage[round]{natbib}
\usepackage{amsthm}
\usepackage{booktabs}
\usepackage{xcolor}
\usepackage[english]{babel}
\newtheorem{theorem}{Theorem}
\usepackage{multirow}
\usepackage[inline]{enumitem}   

\makeatletter
\newcommand{\inlineitem}[1][]{%
\ifnum\enit@type=\tw@
    {\descriptionlabel{#1}}
  \hspace{\labelsep}%
\else
  \ifnum\enit@type=\z@
       \refstepcounter{\@listctr}\fi
    \quad\@itemlabel\hspace{\labelsep}%
\fi}
\makeatother
\title{Mapping poverty at multiple geographical scales}

\author{Silvia De Nicolò\thanks{Department of Statistical Sciences, University of Bologna, Bologna, Italy.} \and Enrico Fabrizi\thanks{DISES, Catholic University of the Sacred Heart, Piacenza, Italy.} \and Aldo Gardini\thanks{\texttt{aldo.gardini@unibo.it}, Department of Statistical Sciences, University of Bologna, Bologna, Italy.}}
\date{}

\begin{document}

\maketitle

\abstract{\noindent Poverty mapping is a powerful tool to study the geography of poverty. The choice of the spatial resolution is central as poverty measures defined at a coarser level may mask their heterogeneity at finer levels. We introduce a small area multi-scale approach integrating survey and remote sensing data that leverages information at different spatial resolutions and accounts for hierarchical dependencies, preserving estimates coherence. We map poverty rates by proposing a Bayesian Beta-based model equipped with a new benchmarking algorithm accounting for the double-bounded support. A simulation study shows the effectiveness of our proposal and an application on Bangladesh is discussed.}\\
\\


\section{Introduction}

 The first goal of the 2030 Agenda for Sustainable Development of the United Nations is to eradicate poverty in all its forms. Given that, current research trends are focusing on the complexity, spatial heterogeneity, and geographical roots of poverty to properly design out-of-poverty paths \citep{allard2017places, fan2021paths, zhou2022geography}. 
 Poverty has a multifaceted nature compounded by its spatial dimension \citep{gauci2005spatial}: as an example, we can mention that widely used classifications divide it into 
urban and rural poverty \citep{christiaensen2014poverty}. Moreover, the study of spatial poverty traps, theorized as a persistent poverty status caused by location characteristics, e.g., remoteness, poor infrastructure and services, 
and/or excessive migration costs \citep{kraay2014poverty} has received considerable attention.

For these reasons, poverty mapping is sparking interest in welfare economics and geography studies \citep{jean2016combining, hall2023review}, as an essential
methodology to support the investigation of the spatial distribution and regional characteristics of poverty.
In addition, if performed at a granular level, it enables a geographical targeting, namely the area-driven allocation of resources for poverty alleviation \citep{elbers2007poverty, liu2017spatio}. This is recognized as an effective and cost-saving tool for poverty reduction if compared to individual targeting as the latter is characterized by high administrative costs of data collection and follow-up. Conversely, regional targeting does not require household tracking, easing the supervision and management processes \citep{bigman2000geographical, galasso2005decentralized}. Furthermore, it enables the diversification of policy instruments employed from area to area that can be easily combined with other antipoverty measures.

When performing poverty mapping, a natural problem relates to the choice of the spatial scale which strictly depends on the purpose of the analysis. This adds up to a more technical problem that arises when aggregating areal data from neighboring zones. Indeed, such an aggregation induces  
a smoothing effect that can mask or distort information on spatial heterogeneity (patchiness) of the target poverty measure (e.g., see Figure \ref{fig:motivating}). This is known in geography as the scaling problem, and represents one of the aspects of the Modifiable Area Unit Problem \citep[MAUP,][]{kolaczyk2001multiscale,pratesi2016introduction}.
 This effect strongly depends on the heterogeneity of the aggregated areas and might lead to misleading conclusions. A simple strategy to deal with it is to undertake the analysis by using multiple scales at once \citep{kolaczyk2001multiscale}.


The choice of the spatial scale is also strongly connected with the availability of data sources.
In particular, surveys on income, consumption and/or living conditions are commonly used in poverty mapping but typically provide reliable estimates at high levels of spatial aggregation (e.g., regional or national).
Survey estimates are increasingly unreliable when considering finer spatial levels, translating into smaller and smaller sample sizes \citep{pratesi2016introduction}. In addition, remote sensing (RS) or mobile phone usage data are recently employed in poverty mapping, especially in developing countries \citep{jean2016combining, schmid2017constructing}. As opposed to survey data, such alternative sources are highly informative at fine spatial levels, while losing power when aggregated. 

\begin{figure}
		\centering
	\includegraphics[width=0.8\textwidth]{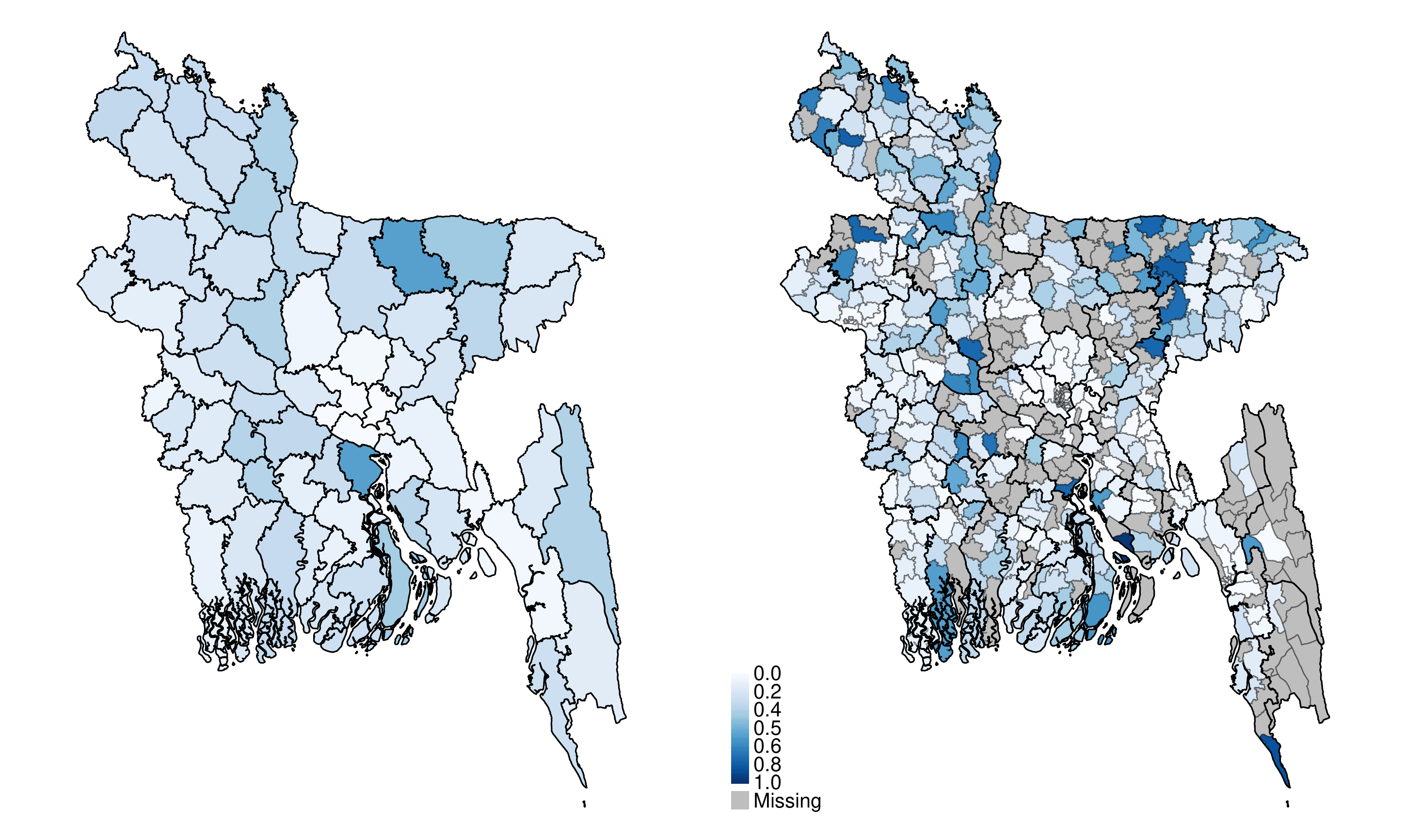}
		\caption{Poverty rate estimates at zila level (district) on the left-hand side and at upazila level (sub-district) on the right-hand side. Note that upazila level heterogeneity is masked when aggregating at the zila level.}
		\label{fig:motivating}
\end{figure}

A good practice in poverty mapping is to promote the efficient integration of multiple sources by means of statistical models \citep{steele2017mapping, zhao2019estimation}. To do so, we have to optimize the informative power that distinct sources have at different scales while accounting for hierarchical dependencies among levels. In this sense, a multi-scale procedure enables such an integration, while overcoming the aforementioned scaling problem. To the best of our knowledge, no explicit formalization of a multi-scale modeling approach has been proposed in the poverty mapping literature. 
Some proposals exploit geo-statistical models to obtain a continuous surface estimation over the study area \citep{puurbalanta2021clipped,sohnesen2022using}, avoiding the scaling problem. Such methods, however, require the exact location of each household which is difficult to obtain due to data disclosure restrictions. 

We propose a Bayesian multi-scale model for poverty mapping that combines survey and alternative sources, available at both finer and coarser spatial levels, and does not require the exact location of survey respondents. The quantity we target is a poverty rate, defined as the proportion of people whose income or living conditions indicator falls below a pre-specified threshold. Our methodology sets in the framework of the Small Area Estimation (SAE) literature, aiming at estimating population quantities in areas (or domains) not originally planned by the survey design \citep[][Ch. 4]{tzavidis2018start,rao2015small}, where area-specific sample sizes are likely to be too small to yield precise enough estimates, or even null. SAE methods have been widely used in poverty mapping \citep{molina2014small, casas2016poverty, corral2022guidelines, bikauskaite2022multivariate}, and most of them rely on models defined at the unit level. Conversely, we adopt an area-level approach, namely integrating the area direct estimates with areal covariates. We remark that this approach requires neither survey microdata nor covariates defined for the entire population.

In this paper, poverty rates are modeled through a Beta-based model \citep{liu2014hierarchical,fabrizi2016bayesian,janicki2020properties}. Specifically, we rely on the Extended Beta proposed in \citet{de2022extended}. Its features solve major issues encountered in poverty mapping such as the double-bounded support of poverty measures, the excess of zero/one values, and the intra-cluster correlation induced by the survey design. A new methodology to assess the uncertainty associated with out-of-sample areas is also proposed. We extend it to deal with multiple spatial scales, enabling us to overcome the scaling problem as well as to efficiently integrate multiple sources; in doing so, we provide a fresh contribution to poverty mapping literature.

Previous proposals that deal with multiple levels in SAE convey information in a single direction: from finer to coarser as the so-called sub-area or two-fold models \citep{torabi2014small, erciulescu2019model}, or from coarser to finer through geo-statistical models \citep{benedetti2022accounting}. Our proposal, instead, leverages the information streams in both directions by considering multiple predictors with shared random effects, introducing in the poverty mapping framework the approach proposed by \citet{aregay2016multiscale, aregay2017comparing} in the disease mapping literature.

A multi-scale approach must ensure the coherence of poverty estimates across levels: poverty rate estimates at a given low-resolution area must correspond to the population-weighted aggregation of higher-resolution area estimates. The misalignment among estimates obtained at multiple levels is settled by the so-called benchmarking techniques \citep{bell2013benchmarking}, reviewed in Section \ref{sec:benchmarking}. Conventional benchmarking methods do not account for the bounded support of poverty rates and may return invalid outcomes, possibly lying below zero or above one. We propose a new benchmarking procedure 
that suitably restricts the parameter space to the unit interval and relies on the idea of posterior projection \citep{dunson2003bayesian, sen2018constrained}. This represents a contribution to the literature as it provides benchmarked estimates and related uncertainty measures that are fully consistent with the bounded support.

Our proposal of multi-scale modeling equipped with a suitable benchmarking procedure is applied to map poverty in Bangladesh at two different spatial layers: zilas, i.e. administrative districts, and upazilas, i.e. sub-districts comparable with counties or boroughs. We use data from the Bangladesh Demographic and Health Survey \citep[DHS;][]{fabic2012systematic} by considering as poverty measure the proportion of people in the area falling in the first quintile of the national distribution of the Wealth Index \citep[WI;][]{pirani2014wealth}. WI is a composite score that summarizes the living conditions of a household and that can be read as a measure of socioeconomic status \citep{poirier2020approaches}.
The insufficient survey sample sizes in most zilas and upazilas (one-third of the upazilas are out-of-sample) drive us to integrate survey with RS data. Open access RS data encompasses night light radiance, population structure and population density, along with geographical characteristics, land use, infrastructures, and other social and economic features known to be poverty predictors at granular levels \citep{zhou2022geography}. Eventually, we compare the performance of our multi-scale model to alternatives by means of a design-based simulation study. 

The paper is organized as follows. Section \ref{sec:notation} introduces the notation; Section \ref{sec:models} is devoted to a detailed description of the  proposed small area model defined at multiple levels, prior specification, and posterior inference algorithms, along with a review of other standard models. The benchmarking proposal is illustrated in Section \ref{sec:benchmarking}. Section \ref{sec:simulation} presents a design-based simulation exercise that tests the performance of our model in comparison with alternatives. An application of poverty mapping in Bangladesh by using DHS and RS data is illustrated in Section \ref{sec:application}. Conclusions are drawn in Section \ref{sec:conclusions}.

\section{Notation and survey-based estimators}\label{sec:notation}

The statistical models discussed afterward primarily rely on survey data. In this section, we introduce the notation and basic survey quantities, such as the poverty rate estimators and related uncertainty measure, which constitutes the starting point of our modeling framework.
Let $\mathcal{U}$ denotes a finite population of $N$ individuals, on which we define two nested partitions, although the following setting can be easily generalized to multiple partitions.
The first partition divides $\mathcal{U}$ into $D$ non-overlapping domains $\mathcal{U}_d$ at a coarser level, labeled as \textit{areas}, each constituted by $N_d$ units. A second partition at a finer level divides each $\mathcal{U}_d$, $d=1, \dots, D$, into $M_d$ \textit{sub-areas} $\mathcal{U}_{dj}$ of size $N_{dj}$, so that $M=\sum_{d=1}^D M_d$ is the overall number of sub-areas. 

A complex survey sample of overall size $n$ is drawn from $\mathcal{U}$ and, on its turn, it can be partitioned into area and sub-area sub-samples of sizes $n_d$ and $n_{dj}$, $\forall d, j$ respectively.
The notations $\hat{\bar {Y}}_{d}$ and $\hat{\bar {Y}}_{dj}$ denote the estimators based on survey data (direct estimators) of a poverty rate for areas $d=1,\dots, D$ and the nested sub-areas $j=1,\dots, M_d$, respectively. 
Let $w_{dji}$ denotes the sampling weight of individual $i=1,\dots,n_{dj}$ pertaining to area $d$ and sub-area $j$, and $y_{dji}$ the indicator variable denoting the poverty status. As our target domains are not planned in the survey design, we employ an H{\'a}jek-type estimator \citep{hajek1971discussion} defined as
\begin{align}
    &\hat{\bar {Y}}_{dj}= \frac{\sum_{i=1}^{n_{dj}}w_{dji}y_{dji}}{\sum_{i=1}^{n_{dj}}w_{dji}},\quad j=1,\dots,M_d,\ d=1,\dots,D \label{direct_suba}\\
    &\hat{\bar {Y}}_{d}= \frac{\sum_{j=1}^{M_{d}}\sum_{i=1}^{n_{dj}}w_{dji}y_{dji}}{\sum_{j=1}^{M_{d}}\sum_{i=1}^{n_{dj}}w_{dji}}=\frac{\sum_{j=1}^{M_{d}}\left(\sum_{i=1}^{n_{dj}}w_{dji}\right)\hat{\bar {Y}}_{dj}}{\sum_{j=1}^{M_{d}}\sum_{i=1}^{n_{dj}}w_{dji}}. \label{direct_area}
\end{align}
Those estimators are approximately unbiased for the unknown population proportions denoted with $\theta_d$ and $\theta_{dj}$, assuming that $0<\theta_d<1$ and $0<\theta_{dj}<1$. 
Defining as $t$ the national poverty rate and as $s_d=N_d/N$ the known population share, it holds that $\sum_{d=1}^D s_d\theta_{d} = t$. Furthermore, the sub-area partition implies 
\begin{align} 
  \sum_{j=1}^{M_d} s_{dj}\theta_{dj}= \theta_{d},\quad d=1,\dots,D,   \label{benchmark1}
\end{align}
where $s_{dj}=N_{dj}/N_d$ with $\sum_d s_d=\sum_j s_{dj}=1$.

Survey estimates from \eqref{direct_suba} and \eqref{direct_area} may be highly imprecise due to the small sample sizes. To assess the reliability of estimates, we need to introduce a measure of uncertainty, included also in the modeling framework of Section \ref{sec:models}.
To account for the complex design, we focus on the effective sample sizes  $\tilde{n}_d$ and $\tilde{n}_{dj}$, i.e. the virtual sample sizes required by a simple random sample to retrieve estimators with standard errors equal to those of $\hat{\bar{Y}}_d$ and $\hat{\bar{Y}}_{dj}$, obtained under complex design. Focusing on the area case, we can express the sampling variance of $\hat{\bar{Y}}_d$ as $\theta_d(1-\theta_d)\tilde{n}_d^{-1}$, being the estimator of a proportion. The equivalent sample sizes incorporate the gains and losses in efficiency due to stratification, weighting, and the possibly strong intra-cluster correlation.  They can be expressed in terms of the design effect DEff$_d$, i.e. the ratio between the variance of the direct estimator and its variance under simple random sampling, as $\tilde{n}_d=n_d/\text{DEff}_d$. The same reasoning applies for $\hat{\bar{Y}}_{dj}$. We remark that, given the different correlation structures occurring within areas and sub-areas, generally $\tilde{n}_d\neq \sum_{j=1}^{M_d}\tilde{n}_{dj}$.

The estimated uncertainty measures allow the detection of spatial levels with unreliable estimates, thus requiring SAE techniques. In our framework, two different spatial levels are considered and we assume that they both present unreliable estimates as surveys are generally planned for national (or regional) aggregates.

\section{Small Area Models}\label{sec:models}

Among SAE techniques, we focus on models defined at the area level. 
Small area models may be also defined at the unit (or individual) level and generally requires the individual auxiliary information to be known for the entire population. As novel data sources, such as RS data, provide areal information, the area-level models seem particularly suitable for poverty mapping.

We consider a hierarchical Bayesian model for proportions based on an extended Beta distribution, introduced in \cite{de2022extended}, by extending it for multiple spatial scales. 
We recall its single-level specification in Section \ref{subsec:extendedbeta}, whereas multi-level specifications follow in Section \ref{subsec:multilevel}.

\subsection{The extended Beta model}
\label{subsec:extendedbeta}

The extended Beta (EB) model is defined as a mixture of a Beta distribution and two Dirac components defined on the boundaries 0 and 1 \citep{fabrizi2016hierarchical,janicki2020properties}. Let us consider the mean-precision parametrization of the Beta random variable \citep{ferrari2004beta} namely, if $Y |\mu, \phi \sim \text{Beta}(\mu \phi, (1-\mu)\phi)$, then its density function can be defined as
\begin{align*}
    f_{B}(y; \mu, \phi)= \frac{\Gamma\left(\phi\right)}{\Gamma\left(\mu\phi\right) \Gamma\left((1-\mu)\phi\right)} y^{\mu\phi-1} (1-y)^{(1-\mu)\phi-1},  \quad y \in (0,1).
\end{align*}
Under such parametrization, $\mu \in (0,1)$ is the expectation and $\phi \in (0, +\infty)$ is a dispersion parameter as $\mathbb{V}[y|\mu, \phi]=\mu(1-\mu)(\phi+1)^{-1}$. This parametrization provides a first intuitive justification for the use of the Beta as a likelihood for proportions, as the first two moments match those of the proportion estimator under simple random sampling.

We use the notation 
$\hat{\bar{Y}}_{dj}|\mu_{dj},\phi_{dj}, \pi_{0dj}, \pi_{1dj} \sim EB(\mu_{dj},\phi_{dj}, \pi_{0dj}, \pi_{1dj})$
to define the extended Beta model for sub-area $dj$, where the mixture is made explicit as follows
\begin{equation}
\begin{aligned}
    &\hat{\bar{Y}}_{dj}|\mu_{dj}, \phi_{dj}, \pi_{0dj}, \pi_{1dj} \sim \pi_{0dj} \times \boldsymbol{1}\{\hat{\bar{Y}}_{dj}=0\}+ \pi_{1dj} \times \boldsymbol{1}\{\hat{\bar{Y}}_{dj}=1\}+\\ 
    &\qquad\qquad+ (1- \pi_{0dj}- \pi_{1dj}) \times \text{Beta}\left(\mu_{dj}\phi_{dj},(1-\mu_{dj})\phi_{dj}\right)\boldsymbol{1}\{\hat{\bar{Y}}_{dj}\in(0,1)\},
\end{aligned}
\label{modello_upa}
\end{equation}
where $\boldsymbol{1}{A}$ is an indicator function assuming value one if $A$ occurs, and zero otherwise. The probabilities of observing 0 and 1 values are denoted with $\pi_{0dj}$ and $\pi_{1dj}$, respectively. As previously mentioned, $\mu_{dj}$ is the location parameter that can be further modeled specifying a linear predictor with possible covariates and random effects. Lastly, in line with most literature on area-level models, the dispersion parameter $\phi_{dj}$ is considered to be known to allow identifiability and replaced by $\tilde{n}_{dj}-1$.
Following \cite{de2022extended}, we assume that direct estimates equal to 0 or 1 are due to a censoring process, as $\theta_{dj} \in (0,1)$ and the observation of boundary values 0 or 1 occurs as a result of scarce sample sizes or strong intra-cluster correlations.

We model $\pi_{0dj}$ and $\pi_{1dj}$ in a parsimonious way, accounting for sample features and probabilistic assumptions \citep{de2022extended}. Firstly, $\mu_{dj}$ is interpreted as the expectation of non-censored observations. Secondly, the observed individual poverty status has a dependency structure which is simplified by assuming an exchangeable dependency among each pair of observations within areas.
Under this assumption, the two probabilities are defined as $\pi_{1dj}=\mu_{dj} \lambda_s^{m_{dj}-1}$ and $\pi_{0dj}=[1+\mu_{dj}(\lambda_s-2)]^{m_{dj}-1}/(1-\mu_{dj})^{m_{dj}-2}$.
The correlation between observations is modeled by $\lambda_s$ which has bounded support: 
$\max \bigg\lbrace 0, \max_{d} \frac{2\mu_{dj}-1}{\mu_{dj}}\bigg\rbrace \leq \lambda_s \leq 1. $

Under the EB model, the population proportion $\theta_{dj}$ is defined by its expectation 
\begin{equation}
\theta_{dj}=\mathbb{E}\left[\hat{\bar{Y}}_{dj}|\mu_{dj}, \lambda_s\right]=\left[1-\mu_{dj} \lambda_s^{m_{dj}-1}-\frac{[1+\mu_{dj}(\lambda_s-2)]^{m_{dj}-1}}{(1-\mu_{dj})^{m_{dj}-2}}\right]\mu_{dj}+\mu_{dj} \lambda_s^{m_{dj}-1}.
\label{theta}
\end{equation}

In the following subsections, the EB is assumed for area-specific direct estimates, coherently adjusting the notation as
$\hat{\bar{Y}}_{d}|\mu_{d}, \phi_d, \pi_{0d}, \pi_{1d} \sim EB(\mu_{d},\phi_{d}, \pi_{0d}, \pi_{1d})$ and with correlation parameter $\lambda_a$.

\subsection{Small area models on multiple spatial scales}
\label{subsec:multilevel}

As mentioned before, poverty has a geographical pattern that varies with respect to the spatial scale of measurement. This feature can be modeled through proper tools known in spatial statistics as multi-scale models. A widespread approach is to account for the scaling effect by exploiting a likelihood factorization holding for Gaussian and Poisson models \citep{kolaczyk2001multiscale}. This is popular in the disease mapping framework, where responses are counts of rare events occurrences \citep{louie2006multiscale}, which is not our case.  Another way to implement multi-scale models is to follow the proposal by \citet{aregay2016multiscale,aregay2017comparing}, building models with distinct linear predictors for each spatial scale that share common random effects. It represents an interesting procedure as it is applicable to any distributional assumption for the response. This is particularly useful since we need Beta-based models to correctly account for the bounded support of poverty rates. For this reason, we extend the approach of \citet{aregay2016multiscale} to deal with poverty rate estimation through the EB likelihood. 

In the area-level models literature, the sole proposal that handles multiple spatial scales is the sub-area (or two-fold) model, usually outlined under Gaussian assumptions  \citep{torabi2014small,erciulescu2019model}. However, it cannot be considered as a proper multi-scale model since only sub-area direct estimates and associated uncertainty measures are employed. In this way, the survey information available at the coarser level is discarded, even if more reliable. 
In contrast, our modeling proposal, labeled as shared multi-scale model, employs survey data at both levels and is presented in Section \ref{sec:shared}. Sub-area model with EB likelihood is outlined in Section \ref{sec:subarea} and, in Section \ref{sec:indep} we define as a benchmark a multi-scale model without shared random effects, labeled as independent multi-scale model.

\subsubsection{Shared multi-scale model}\label{sec:shared}

The proposed shared multi-scale (S-MS) model is defined on two distinct spatial levels:
\begin{align}
    &\hat{\bar{Y}}_{d}|\mu_{d}, \phi_{d},  \pi_{0d}, \pi_{1d} \stackrel{ind}{\sim} EB(\mu_{d},\phi_{d}, \pi_{0d}, \pi_{1d}),\quad d=1,\dots,D,\nonumber\\
        & \text{logit}\left(\mu_{d}\right)=\alpha_a+\mathbf{x}^T_{d}\boldsymbol\beta_a+u_{d};\label{eq:mod_shared1}\\
        &\hat{\bar{Y}}_{dj}|\mu_{dj}, \phi_{dj}, \pi_{0dj}, \pi_{1dj} \stackrel{ind}{\sim} EB(\mu_{dj},\phi_{dj}, \pi_{0dj}, \pi_{1dj}),\quad j=1,\dots M_d,\ d=1,\dots,D,\nonumber\\
        & \text{logit}\left(\mu_{dj}\right)=\alpha_s+\mathbf{x}^T_{dj}\boldsymbol\beta_s+v_{dj}+u_d;\label{eq:mod_shared2}
\end{align}
where $\mathbf{x}_{d}$ and $\mathbf{x}_{dj}$ are vectors of area and sub-area covariates. The corresponding coefficients $\boldsymbol\beta_a$ and $\boldsymbol\beta_s$ differ between predictors \eqref{eq:mod_shared1} and \eqref{eq:mod_shared2} to account for changes in the functional relationship induced by the spatial scale (ecological bias).
Note that such a model promotes an exchange of information between levels through the use of the shared area-specific random effect $u_d$. In such a way, the information nested at finer levels directly contributes to poverty estimation at coarser levels, enabling the communication flow to be two-ways. Moreover, $u_d$ accounts for the correlation between levels.

The sub-area population proportion $\theta_{dj}$ can be obtained as a function of $\mu_{dj}$ and $\lambda_s$ according to \eqref{theta} and the area proportion $\theta_d$ is defined as a linear combination of $\theta_{dj}$ as in \eqref{benchmark1}, exploiting the hierarchical partition. Note that, alternatively, $\theta_d$ can be expressed in terms of the model expectation $\mathbb{E}\left[\hat{\bar{Y}}_{d}|\mu_{d}, \lambda_a \right]$, obtained adapting \eqref{theta} to the area layer. We remark that this strategy would not preserve the coherence between poverty rates at different levels as opposed to our strategy.
Another advantage of using \eqref{benchmark1} is that the peculiarities related to possible out-of-sample sub-areas are taken into account by incorporating their model predictions relying on available auxiliary information (see Section \ref{post_inference}). In contrast, $\mathbb{E}\left[\hat{\bar{Y}}_{d}|\mu_{d}, \lambda_a \right]$ is informed only by direct estimates and covariates at the area level. This helps mitigating the possible informativeness of the process that leads sub-areas out of the sample. Actually, it can be the effect of the unplanned allocation of sample units within areas or it may depends on sub-area characteristics such as difficult access, remoteness or other.

\subsubsection{Sub-area models}\label{sec:subarea}

The sub-area (SA) model defined with the EB likelihood is defined as: 
\begin{equation}\label{eq:mod_subarea}
    \begin{aligned}
        &\hat{\bar{Y}}_{dj}|\mu_{dj}, \phi_{dj}, \pi_{0dj}, \pi_{1dj} \stackrel{ind}{\sim} EB(\mu_{dj},\phi_{dj}, \pi_{0dj}, \pi_{1dj}),\quad j=1,\dots M_d,\ d=1,\dots,D;\\
        & \text{logit}\left(\mu_{dj}\right)=\alpha_s+\mathbf{x}^T_{dj}\boldsymbol\beta_s+v_{dj}+u_{d}.
    \end{aligned}
\end{equation}
Such a model account for correlation within areas through $u_d$, but it ignores the multiple scales of the problem since $u_d$ is not informed by any quantity defined at the coarser level. Similarly to the S-MS model, the $\theta_{dj}$ can be estimated as a function of $\mu_{dj}$ and $\lambda_s$ according to \eqref{theta}, while $\theta_d$ is defined exploiting the hierarchical partition as in \eqref{benchmark1}.

\subsubsection{Independent multi-scale model}\label{sec:indep}

A naive approach to produce estimates at distinct levels of disaggregation would be fitting a model with two independent layers. In this way, the auxiliary information at a given level is combined to corresponding direct estimates, determining a multi-scale modeling procedure. The independent multi-scale (I-MS) model is defined as follows:
\begin{equation}
    \begin{aligned}\label{eq:mod_ind}
    &\hat{\bar{Y}}_{d}|\mu_{d}, \phi_{d},  \pi_{0d}, \pi_{1d} \stackrel{ind}{\sim} EB(\mu_{d},\phi_{d}, \pi_{0d}, \pi_{1d}),\quad d=1,\dots,D;\\
        & \text{logit}\left(\mu_{d}\right)=\alpha_a+\mathbf{x}^T_{d}\boldsymbol\beta_a+u_{d}\\
        &\hat{\bar{Y}}_{dj}|\mu_{dj}, \phi_{dj}, \pi_{0dj}, \pi_{1dj} \stackrel{ind}{\sim} EB(\mu_{dj},\phi_{dj}, \pi_{0dj}, \pi_{1dj}),\quad j=1,\dots M_d,\ d=1,\dots,D;\\
        & \text{logit}\left(\mu_{dj}\right)=\alpha_s+\mathbf{x}^T_{dj}\boldsymbol\beta_s+v_{dj}.
    \end{aligned}
\end{equation}
The population area proportions are defined as $\theta_d=\mathbb{E}\left[\hat{\bar{Y}}_{d}|\mu_{d}, \lambda_a \right]$ and $\theta_{dj}=\mathbb{E}\left[\hat{\bar{Y}}_{dj}|\mu_{dj}, \lambda_s \right]$ in line with \eqref{theta}. As a consequence, the coherency between rates at multiple levels is not preserved. In addition, the linear predictor related to the finer level in \eqref{eq:mod_ind} does not present an area-specific random effect to model the hierarchical dependence between layers.

\subsection{Prior distributions}

The prior distributions for model parameters are in line with those proposed in \citet{de2022extended}, which can be considered for further details. All three model settings received the same priors for the corresponding parameters. The regularized horseshoe prior \citep{horseshoe} is adopted for regression coefficients in order to automatically incorporate the variable selection step, mimicking the behavior of a spike-and-slab prior.
Both the sets of area and sub-area specific random effects present a variance-gamma shrinkage prior. To complete the model, a uniform prior is specified for the correlation parameter
$\lambda_a | \mu_1, \dots, \mu_D$ $\sim \text{Unif}\left(\max \lbrace 0, \max_{d} (2\mu_{d}-1)/\mu_d\rbrace ; 1 \right)$, and the same for $\lambda_s$ with $\mu_{d}$ replaced by $\mu_{dj}$.

\subsection{Posterior Inference}
\label{post_inference}

Posterior inference has been performed via
Markov Chain Monte Carlo (MCMC) techniques. This has been implemented through the \texttt{Stan} language and the \texttt{rstan} package \citep{carpenter2017stan}.
The estimation has been performed with 4 chains, each including 4,000 iterations, where the first 2,000 has been considered as warm-up iterations and discarded. 
	The expectations of $\theta_{dj}| \mathbf{y}$ and $\theta_{dj}| \mathbf{y}$ are adopted as point estimators, labeled as model-based or hierarchical Bayes (HB) estimators: $
  \hat{\theta}_{d} =\mathbb{E}[\theta_{d}|\mathbf{y}]$ and $\hat{\theta}_{dj}=\mathbb{E}[\theta_{dj}|\mathbf{y}]$.
  Their estimates, together with other posterior summaries such as credible intervals and posterior variances, can be easily approximated exploiting MCMC draws. Note that $\theta_d | \mathbf{y}$, if defined through \eqref{benchmark1}, can be retrieved by combining draws from $\theta_{dj} | \mathbf{y}$. 
The HB estimators under the EB model enjoy the property of design-consistency, namely conditioning on higher level parameters $\hat{\theta}_d \stackrel{p}{\rightarrow} \hat{\bar{Y}}_d$ and $\hat{\theta}_{dj} \stackrel{p}{\rightarrow} \hat{\bar{Y}}_{dj}$ \citep[]{de2022extended}.

   As regards out-of-sample sub-areas, the prediction is carried out by considering the functional
	$
\theta_{dj}^{OOS}=\mu_{dj}=\text{logit}^{-1}\left(\mathbf{x}_{dj}^T\boldsymbol{\beta}_s+v_{dj}+u_d\right),
	$
 when auxiliary information is available.
	Samples from $\theta_{dj}^{OOS}|\mathbf{y}$ are obtained combining draws from  $\boldsymbol{\beta}_s|\mathbf{y}$ and $u_d|\mathbf{y}$. As $v_{dj}$ constitutes a random effect from an unobserved sub-area, we propagate the uncertainty by drawing samples from the variance-gamma shrinkage priors \citep[][]{de2022extended}.

\section{The Bayesian benchmarking proposal}
\label{sec:benchmarking}

The definition of area level estimates through \eqref{benchmark1} preserves 
the coherency of poverty rates between the two levels. However, the coherency with respect to higher levels of aggregation (e.g., national) is not guaranteed. Indeed, the poverty rate at higher levels may be known in population or reliably estimated through surveys; such value, hereafter called benchmark, might be used to constrain the linear combination of sub-area estimates.
In this section, we consider the problem of finding a set of constrained estimators $\tilde{\boldsymbol\theta}=(\tilde{\theta}_{11}, \dots \tilde{\theta}_{DM_d})^T$, for the target parameters $\boldsymbol\theta=(\theta_{11},  \dots \theta_{DM_d})^T$. The linear constraint is imposed by the poverty rate at the finest spatial scales and it is defined by the combination of both equations in \eqref{benchmark1} as
\begin{eqnarray} \label{bmk}
  \sum_{d=1}^D \sum_{j=1}^{M_d}  q_{dj}\tilde{\theta}_{dj}&=&  t,
\end{eqnarray}
where $q_{dj}=s_{dj}s_d$ denotes the population share of sub-area $dj$ at national level. Note that $\sum_d\sum_j q_{dj}=1$, and $t$ is, in our case, the poverty rate at the national level.

The problem of benchmarking has been tackled from different perspectives in the Bayesian literature. The most popular strategy follows a decision-theoretic framework: starting from a loss function $L(\boldsymbol\theta, \tilde{\boldsymbol\theta})$, the benchmarked estimators are obtained by minimizing the posterior risk $\mathbb{E}[L(\boldsymbol\theta, \tilde{\boldsymbol\theta})|\mathbf{y}]$ under linear constraint \citep{datta2011bayesian}. Despite its simple implementation, this method is applied to point estimators obtaining only a set of benchmarked point estimators. Thus, it represents a non-fully Bayesian strategy, not delivering posterior measures of uncertainty.
To solve this problem, fully Bayesian methods have been developed. Among the others, \citet{zhang2020fully} include the constraint through a suitable prior distribution; 
\citet{janicki2017benchmarking} search for a constrained joint posterior distribution with minimum Kullback-Leibler distance to the unconstrained one, whereas \citet{okonek2022computationally} opt for an importance sampling approach.

Our approach relies on the concept of posterior projection \citep{dunson2003bayesian, sen2018constrained} that was already considered by \citet{patra2019constrained} within the SAE framework.
In this case, the posterior samples from $\theta_{dj} | \mathbf{y}$ are projected into the feasible set
defined by the benchmarking constraint, minimizing the distance from the original ones. Specifically, this distance is defined by means of a loss function. In contrast to other fully Bayesian methods, such an approach is more computationally efficient, in particular when the fitted model is not trivial as in our case.  
The choice of a suitable loss function for our inferential problem is discussed in Section \ref{subsec::choice_loss} and the main results are illustrated in Section \ref{subsec::projection}.

\subsection{Choice of the loss function}
\label{subsec::choice_loss}

The mapping required for posterior projection can be represented by a loss function $L(\boldsymbol\theta, \tilde{\boldsymbol\theta})$.
Defining with $\psi_{dj}$ a generic weight, the most common loss function is the weighted quadratic one, i.e. $L(\boldsymbol\theta, \tilde{\boldsymbol\theta})=\sum_{d=1}^D \sum_{j=1}^{M_d} \psi_{dj} ({\theta}_{dj}-\tilde{\theta}_{dj})^2$.
This function is suitable when each $\theta_{dj}$ is defined on the real line but not appropriate when it is bounded, as in our case where $\boldsymbol\theta \in \Theta = (0,1)^M$. Coherently, the constrained parameter space has to be bounded too, being defined as
$\Tilde{\Theta}=\{\tilde{\boldsymbol\theta} \in (0,1)^M| \sum_d \sum_j q_{dj}\tilde{\theta}_{dj} = t\}$. 
A similar issue is considered by \citet{ghosh2015benchmarked} which propose a generalized Kullback-Leibler loss function that restricts the parameter space to positive-defined values and avoids estimates falling below 0. A possible adaptation to the case in $(0,1)$ has been proposed by \citet{aitchison1992criteria} contemplating the simple logit transformation, i.e. $\sum_{j=1}^{M_d} \psi_{dj} (\text{logit}({\theta}_{dj})-\text{logit}(\tilde\theta_{dj}))^2$. Such an option, while being intuitively simple, leads to a non-convex function. 

We propose a weighted loss function pertaining to the Bregman family \citep{banerjee2005optimality} that suitably restricts both $\boldsymbol\theta$ and $\tilde{\boldsymbol\theta}$ to lie on the unit interval, defined as follows
\begin{equation}\label{eq::bregman}
L(\boldsymbol\theta, \tilde{\boldsymbol\theta})= \sum_{d=1}^D \sum_{j=1}^{M_d} \psi_{dj} \bigg[ \theta_{dj}\log \bigg( \frac{\theta_{dj}}{\tilde{\theta}_{dj}} \bigg)+(1-\theta_{dj})\log \bigg( \frac{1-\theta_{dj}}{1-\tilde{\theta}_{dj}} \bigg)\bigg].
\end{equation}
In this way, we target all the solutions in $\Tilde{\Theta}$ which minimize $L(\boldsymbol\theta, \tilde{\boldsymbol\theta})|\mathbf{y}$. 
The class of Bregman loss functions includes a number of functions defined over different domains, including the quadratic and the Kullback-Leibler loss functions considered by \citet{ghosh2015benchmarked}. 
Such a class of functions is appealing since it ensures that the posterior risk $\mathbb{E}[L(\boldsymbol\theta, \tilde{\boldsymbol\theta})| \mathbf{y}] $ is minimized by the posterior expectation $\mathbb{E}[\boldsymbol{\theta}|\mathbf{y}]$, a popular Bayes estimator in the small area literature. 
Figure S1 in the Supplementary Material displays the Bregman loss function in comparison with the quadratic loss function. Note that the proposed function is strictly defined in $(0,1)$, being symmetric only when $\theta=0.50$.

To fully define a weighted loss function, its weights must be specified by the user. Popular choices for $\psi_{dj}$ are generally $q_{dj}$, $q_{dj}/\hat{\theta}_{dj}$, or an inversely proportional function of the sampling variances and/or the posterior variances of target parameters \citep[Ch. 6]{rao2015small}.
In our case, we opt for the conservative option $\psi_{dj}=q_{dj}$, so that all estimates are adjusted depending only on their distance to the boundaries of the support.

\subsection{Benchmarked posterior projection}
\label{subsec::projection}

    The benchmarking strategy relying on the posterior projection approach is based on the idea of drawing samples from the unconstrained posterior $p(\boldsymbol{\theta}|\mathbf{y})$, defined on $\Theta$, and projecting them on the constrained parameter space $\tilde{\Theta}$, inducing a mapping function $f:\Theta \rightarrow{} \Tilde{\Theta}$. The following theorem determines the projection mapping that must be applied to the unconstrained estimators.

\begin{theorem}
\label{theorem1}

The projection of $\boldsymbol{\theta}|\mathbf{y}$ on $\tilde\Theta\subseteq(0,1)^M$ induced by the minimum distance mapping based on \eqref{eq::bregman} is a random variable denoted with $\boldsymbol{\tilde\theta}|\mathbf{y}$. The induced transformation is unique and given by:
\begin{equation}  
\tilde{\theta}_{dj} | \mathbf{y}=\frac{1}{2}+\frac{\psi_{dj}}{2\gamma q_{dj}} \left\{  \sqrt{\bigg(1-\gamma \frac{q_{dj}}{\psi_{dj}}\bigg)^2+4\theta_{dj} \gamma\frac{ q_{dj}}{\psi_{dj}}} -1 \right\} \bigg| \mathbf{y}
\label{eq::breg1}
\end{equation}
where $\gamma$ is given as the solution of the non-linear equation
\begin{equation}  
\frac{1}{2\gamma}  \left[  \sum_{d=1}^D \sum_{j=1}^{M_d} \psi_{dj}  \left\{ \sqrt{ \bigg( 1-\frac{\gamma q_{dj}}{\psi_{dj}}\bigg)^2 +4\gamma\theta_{dj} \frac{ q_{dj}}{\psi_{dj}}} -  1 \right\} + \gamma q_{dj} \right]= t.
\label{eq::breg2}
\end{equation}
\end{theorem}

\begin{proof} See the Supplementary Material.
\end{proof}

When inference is carried out via MCMC methods, drawns from $p(\boldsymbol{\tilde\theta}|\mathbf{y})$, i.e. the projected posterior,
can be easily obtained by applying the transformation defined in Theorem \ref{theorem1} to posterior samples from 
$p(\boldsymbol\theta|\mathbf{y})$. The related posterior summaries can be computed as described in Section \ref{post_inference}, with point estimators denoted with $\hat{\tilde{\theta}}_{d}$ and $\hat{\tilde{\theta}}_{dj}$.

As mentioned before, non-fully benchmarking methods target the constrained minimization of a posterior risk, i.e. $\mathbb{E}[L(\boldsymbol\theta, \tilde{\boldsymbol\theta})|\mathbf{y}]$.
Note that our projection approach targets the constrained minimization of $L(\boldsymbol\theta, \tilde{\boldsymbol\theta})|\mathbf{y}$, instead. 
One could argue that the conceptual approach is distinct and may lead to different results with respect to the standard posterior risk minimization, inducing non-comparability issues. Due to the linearity property of the derivative of the Lagrangian function with respect to $\theta_{dj}$, the constrained minimization of the posterior risk leads to a solution that is identical to \eqref{eq::breg1} and \eqref{eq::breg2} but with $\theta_{dj}|\mathbf{y}$ replaced by $\mathbb{E}[\theta_{dj} | \mathbf{y}]$. Thus, Jensen's inequality explains the discrepancy between the two results.
When $\psi_{dj}=q_{dj}$, the mapping function $f:\Theta \rightarrow{} \Tilde{\Theta}$ may be either slightly concave or slightly convex depending on the sign and magnitude of $t-\sum_{d}\sum_{j} q_{dj}\theta_{dj}$
as visually illustrated in Figure S2 of the Supplementary Material:
in this case, the discrepancy is negligible. In contrast, if $\psi_{dj} \neq q_{dj}$, the mapping function may be highly non-linear and results could be incomparable.

\section{Simulation}\label{sec:simulation}
In this section, we study how statistical models of Section \ref{sec:models} deal with a multi-scale data problem by evaluating the frequentist properties of resulting estimators. To this aim, we reproduce a framework in which the goal is making inference on a proportion and, hence, a dichotomous response is generated for the whole synthetic population. We opt for a design-based simulation study, being able to reproduce the peculiarities of a multi-scale problem within a survey design setting. To this aim, we consider two aggregation levels, i.e., areas and sub-areas, and two different scenarios that contemplate diverse ratios of between and within area variability.

For each scenario, a population of $N=180,000$ individuals grouped into $3,600$ clusters is generated, mimicking the two-stages structure characterizing the scheme of several common surveys.
Units are divided into $D=30$ areas and $M=150$ sub-areas, i.e., each area has $M_d=5$ sub-areas. The whole procedure for generating the population and sample drawing is fully detailed in Section S3 of the Supplementary Material. A continuous variable is generated from a three-fold Gaussian model with cluster effect scale $\sigma_c=0.2$, area effect scale $\sigma_a$ and sub-area effect scale $\sigma_s$. In scenario 1, the sub-area variation ($\sigma_s = 0.13$) prevails on the area variation ($\sigma_a = 0.08$), inducing a marked heterogeneity within the areas. In scenario 2, the area variation ($\sigma_a = 0.13$) prevails on sub-area variation ($\sigma_s=0.08$), leading to homogeneous areas. 
Such values for the scale parameters are retrieved from the variance decomposition related to real data discussed in Section \ref{sec:application}.
Lastly, the dichotomous response is obtained by assigning value 1 to individuals below the first quintile of the generated continuous variable, and 0 otherwise.

\begin{figure}
    \centering
    \includegraphics[width = \linewidth]{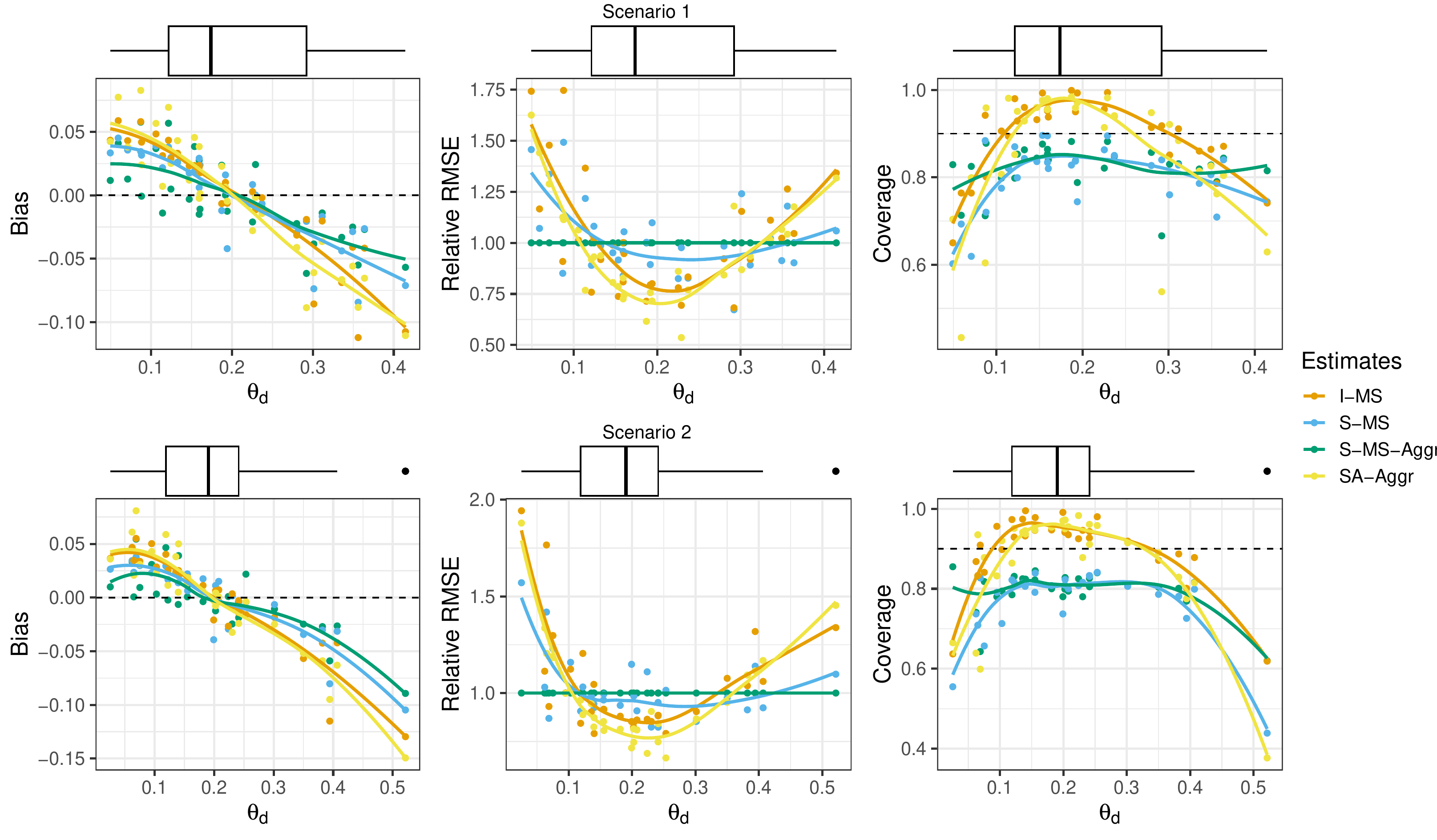}
    \caption{Ensamble distribution of biases, relative RMSEs (i.e. divided by RMSEs of S-MS-Aggr) and frequentist coverages with respect to the population poverty rate $\theta_d$.}
    \label{fig:ris_sim}
\end{figure}

From each synthetic population, $B=1,000$ samples are drawn following a two-stages sampling procedure with rates of approximately 2\%. To account for the presence of out-of-sample sub-areas, each drawn sample does not include units from 37 sub-areas belonging to 15 areas. In this way, half of the areas include individuals from all the sub-areas, and, in the other half, units pertaining to two or three sub-areas are not observed. 

\begin{table}[]
\centering
\begin{tabular}{@{}lllllllll@{}}
\toprule
                   &           & LOOIC   & \multicolumn{3}{c}{MAPE (\%)}    & \multicolumn{3}{c}{ARRMSE}   \\ \cmidrule(l){3-3} \cmidrule(l){4-6} \cmidrule(l){7-9} 
Scen.           & Estimates     & Overall & No OOS & OOS & Overall & No OOS & OOS & Overall \\ \midrule
\multirow{4}{*}{1} & I-MS      & -17.3   & 33.7  & 45.5    & 39.6   & 0.378  & 0.440    & 0.409   \\
                   & S-MS      & -29.5   & 28.2  & 36.4    & 32.3   & 0.351  & 0.465    & 0.408   \\
                   & S-MS-Aggr & -       & 25.8  & 36.2    & 31.0   & 0.323  & 0.468    & 0.395   \\
                   & SA-Aggr   & -       & 27.8  & 40.3    & 34.0   & 0.334  & 0.463    & 0.398   \\ \midrule
\multirow{4}{*}{2} & I-MS      & -17.1   & 50.1  & 50.5    & 50.3   & 0.426  & 0.454    & 0.440   \\
                   & S-MS      & -33.5   & 30.5  & 36.9    & 33.7   & 0.386  & 0.472    & 0.429   \\
                   & S-MS-Aggr & -       & 27.4  & 36.9    & 32.1   & 0.347  & 0.478    & 0.413   \\
                   & SA-Aggr   & -       & 31.6  & 39.4    & 35.5   & 0.386  & 0.473    & 0.429   \\ \bottomrule
\end{tabular}
\vspace{2mm}
\caption{LOOIC, MAPE and ARRMSE for the area level. Concerning MAPE and ARRMSE, values obtained for the areas without missing sub-areas (No OOS), with missing sub-areas (OOS) and overall are distinctly reported.}
\label{tab:sim_area}
\end{table}

The frequentist properties of the benchmarked estimators $\hat{\tilde{\theta}}_{dj}$ and $\hat{\tilde{\theta}}_d$ obtained under different models are compared: S-MS, SA and I-MS. Note that, with the suffix -Aggr we denote area-level estimates defined as in \eqref{benchmark1}.
Let us indicate with $\theta_{i}$ as the true value for the generic area or sub-area $i$ and $\hat{\tilde{\theta}}_{i}^{(b)}$ as the benchmarked estimate at iteration $b$. To determine the 90\% credible intervals we consider 5$^{th}$ and $95^{th}$ posterior percentiles, labeled as $\tilde{\theta}_{i,L}^{(b)}$ and $\tilde{\theta}_{i,U}^{(b)}$. We define bias, root mean squared error (RMSE) and coverage for 90\% credible intervals for each area or sub-area as
\begin{equation*}
    \begin{aligned}
&\text{Bias}\left[\hat{\tilde\theta}_{i}\right]=\frac{1}{B}\sum_{b=1}^B\left(\hat{\tilde{\theta}}_{i}^{(b)}-\theta_{i}\right),\quad \text{RMSE}\left[\hat{\tilde\theta}_{i}\right]=\sqrt{\frac{1}{B}\sum_{b=1}^B\left(\hat{\tilde{\theta}}_{i}^{(b)}-\theta_{i}\right)^2},\\
&\text{Cov}\left[\tilde\theta_{i}\right]=\frac{1}{B}\sum_{b=1}^B\boldsymbol{1}\left\{\theta_{i}\in\left[\tilde{\theta}_{i,L}^{(b)};\tilde{\theta}_{i,U}^{(b)}\right]\right\}.
    \end{aligned}
\end{equation*}

 As summary measures of bias and RMSE, we define the mean absolute percentage mean squared error (MAPE) and the average relative root mean squared error (ARRMSE) as
\begin{equation*}
    \text{MAPE (\%)}=\frac{100}{BI}\sum_{b=1}^B\sum_{i=1}^I\frac{\hat{\tilde{\theta}}_{i}^{(b)}-\theta_{i}}{\theta_{i}}, \quad \text{ARRMSE}=\frac{1}{I}\sum_{i=1}^I\frac{\text{RMSE}\left[\hat{\tilde\theta}_{i}\right]}{\theta_{i}}.
\end{equation*}
Lastly, as a measure of the model goodness-of-fit, the leave-one-out information criterion \citep[LOOIC,][]{vehtari2017practical} is considered.

Interesting cues about the performances of the compared area-specific estimators can be deduced from Figure \ref{fig:ris_sim}. It reports the relative RMSEs, obtained by dividing each RMSE by the corresponding S-MS-Aggr RMSE, jointly with biases and frequentist coverages versus the population values. In both scenarios, the S-MS-Aggr estimators are characterized by lower biases, if compared to I-MS and SA cases: such differences become relevant for areas whose true values are far from their average (about 50\% of them). Such a decrease in the bias leads to some differences also in the behaviour of RMSE too: S-MS-Aggr produces more efficient estimators in the tails of the \textit{ensamble} distribution of population poverty rates $\{\theta_1,\dots,\theta_D\}$, becoming less efficient in the central part of the distribution. By focusing on the 90\% credible intervals, we observe that S-MS-Aggr is affected by a slight under-coverage (median: 0.83 in scenario 1 and 0.81 in scenario 2). However, the performances deteriorate more slowly if compared to the other methods in the non-central parts of the ensamble distribution. Table \ref{tab:sim_area} summarises the overall performances of the methods across areas. The MAPE points out the ability of S-MS-Aggr in reducing the bias, registering lower values both in areas with and without out-of-sample sub-areas; conversely, ARRMSE highlights that the average performances are quite similar.

We can point out that the I-MS model, characterized by the absence of an area-specific random effect, leads to poorer performances if compared to all the other strategies. This is more evident in Scenario 2 where the data-generating process mimics the case $\sigma_a>\sigma_s$. 
Targeting the estimation of sub-area parameters, summaries reported in Table S1 of the Supplementary Material point out that the three models behave quite similarly, with the exception of I-MS, which still shows lower performances in Scenario 2. 
Lastly, we highlight that the average LOOIC of the S-MS model is widely lower both for sub-area 
(Table S1) and area (Table \ref{tab:sim_area}) levels in both scenarios, denoting a remarkable improvement in the predictive abilities of this model.  

In summary, the estimators under the S-MS model have more stable performances along the support. They are characterized by a more moderate shrinkage and, for this reason, they perform better on the tails of the poverty rates ensemble distribution, and slightly better than their competitors on average. Small and, especially, large poverty rates are relevant for the implementation of geographically targeted alleviation policies as they can identify poverty hotspots. Therefore, an improvement in the estimation of such values may be of great importance in this context.

\section{Poverty mapping in Bangladesh}\label{sec:application}

The aim of this section is to map poverty in Bangladesh at multiple levels of administrative divisions: the Administrative Level AL-2, comprising $D=64$ zilas as areas, and the AL-3 with $M=544$ upazilas as sub-areas. 
The target poverty measure is the proportion of people below the 20th percentile of the WI national distribution. We consider DHS survey data complemented with RS and geographical data available from multiple open sources. Data and sources are shortly illustrated in Section \ref{data_survey}, whereas results are discussed in Section \ref{sec:results}.

\subsection{Data}
\label{data_survey}

The Bangladesh DHS, of which we consider the 2014 wave, targets the population residing in non-institutional dwelling units by providing a two-stage stratified sample with 17,300 households and $n=81,624$ individuals overall.
The sample sizes in terms of individuals range from 135 to 4,730 (median 986) for zilas and from 16 to 1,884 (median: 160) for upazilas. All the zilas are in-sample, whereas about one third of upazilas (179) are not included in the sample. 

For each household, a WI score is computed by combining a set of answers on the availability of durable assets and housing welfare characteristics \citep{fabic2012systematic}. Households laying in the first quintile of the national distribution of WI are labeled as \textit{poor}. Since poverty is usually investigated individually, the analysis is carried out at the individual level by assuming that all components of the same household share the same WI score.
A sampling weight is associated to each household, accounting for unequal probabilities of selection and non-responses.
The survey estimates, as defined in Section \ref{sec:notation}, range from 0 to 0.53 (median: 0.21) with just one 0 value for zilas; while they span from 0 to 0.96 (median: 0.16) with 66 zero values for upazilas. 

To retrieve the effective sample sizes, we opt to estimate the design effect starting from the strata level as in \citet{schmid2017constructing}.
Such uncertainty estimates confirm the unreliability of survey estimates both at zila and upazila levels motivating us to employ SAE: we observe coefficients of variation higher than 0.20 in more than 80\% of overall in-sample estimates. We note that effective sample sizes are much lower than the effective ones because the poverty status is defined at the household level and the correlation within clusters (i.e. Census enumeration areas) is very strong.

\begin{table}[]
\centering
\begin{tabular}{@{}rrrr@{}}
\toprule
& \multicolumn{3}{c}{LOOIC (Standard Error)}       \\ 
Models     & SA & I-MS & S-MS\\
\toprule
Zila level &     - \hspace{4mm} (-)&   -137.9        (12.5) &    -150.9     (10.9)\\ 
Upazila level &    -101.8      (30.4)  &       -98.4       (31.2) &     -115.9     (29.5)  \\ \bottomrule
\end{tabular}
\vspace{2mm}
\caption{LOOIC values and their standard errors for all the considered models.}
\label{tab::looic}
\end{table}
  
We integrate $46$ RS variables as covariates computed both at the area and sub-area levels by cutting rasters through shapefiles \citep[for more details, see][and Section S4 in the Supplementary Material]{de2022extended}. Specifically, we consider the demographic composition of areas and sub-areas by including the population density and its disaggregation by age and sex classes. The nighttime light radiance and the distances to main facilities and infrastructures (e.g. nearest city or healthcare site) are included as indicators of area urbanization and economic development. 
As the agricultural sector is a driving force for Bangladeshi economy, its territorial and climatic characterization may be proxies of productivity.
The first one has been considered by including land-use variables that specify the distance from nearest areas with specific use classification (e.g. cultivated, woody-tree, artificial surface), the elevation above sea level and topographic slope.
The second one is fulfilled through bio-climatic variables such as the overall and seasonal variations of temperature and rainfall (e.g., annual mean, standard deviation and temperature diurnal range).

\subsection{Results}\label{sec:results}

The alternative models SA, I-MS and S-MS are fitted on Bangladeshi data. Table \ref{tab::looic} provides details about the model comparison in terms of LOOIC, for both spatial levels. According to it, the S-MS model shows better goodness-of-fit with respect to the other proposals, while SA and I-MS perform similarly.
The main motivation for this finding is that the S-MS model has the area-specific random effect $u_d$, $\forall d$, characterizing the distributions at both levels, thus promoting an exchange of information between areas and sub-areas. Conversely, the other models define $u_d$ at only one level, thereby being informed only by observations at the corresponding level. This translates into a better ability of S-MS model to capture specific area features through random effects. To elaborate more on this point, the ensemble distributions of $\mathbb{E}[u_d|\mathbf{y}]$ are depicted in the left panel of Figure~\ref{fig:dens_reff}: note that under SA and I-MS models such distributions are markedly more concentrated around zero, not capturing area-specific peculiarities. Looking at the ensemble distribution of $\mathbb{E}[v_{dj}|\mathbf{y}]$, shown in the right panel of Figure~\ref{fig:dens_reff}, we observe that upazila effects seem partially absorbed by zila effects for the S-MS model. However, no relevant differences among models are recorded in this case. 

\begin{figure}
    \centering
    \includegraphics[width = .8\linewidth]{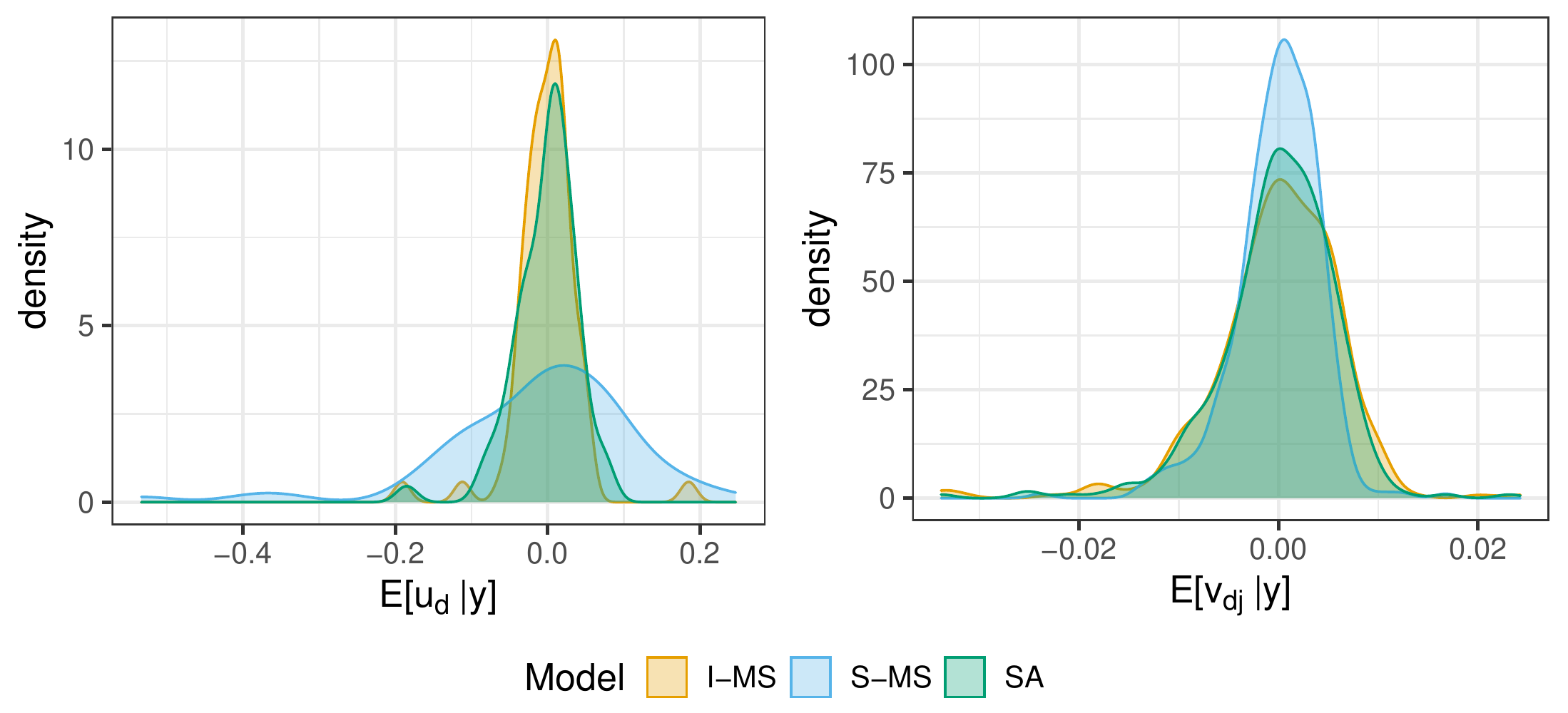}
    \caption{Densities of posterior means of  $u_d$ and $v_{dj}$, $\forall d$, $\forall j$, for each model.} 
    \label{fig:dens_reff}
\end{figure}

To simplify discussion, from now on, we focus on S-MS model estimates due to the aforementioned reasons. 
To obtain benchmarked estimates with respect to the national level, we project the posterior distributions of $\theta_{dj}$, $\forall d$, $\forall j$ exploiting the result of Theorem 1. Exploiting the definition of the response, we set $t=0.2$.
We compare the original posterior distributions with their projections under the proposed Bregman and weighted squared error loss functions. In particular, Figure \ref{fig:bench_sforo} shows it for two upazilas with poverty rates close to zero. Note that, under the weighted squared error loss, projections can be markedly outside the support taking negative values. In contrast, under the Bregman loss function, benchmarked posteriors are correctly defined in $(0,1)$. 

\begin{figure}
    \centering
    \includegraphics[width = .8\linewidth]{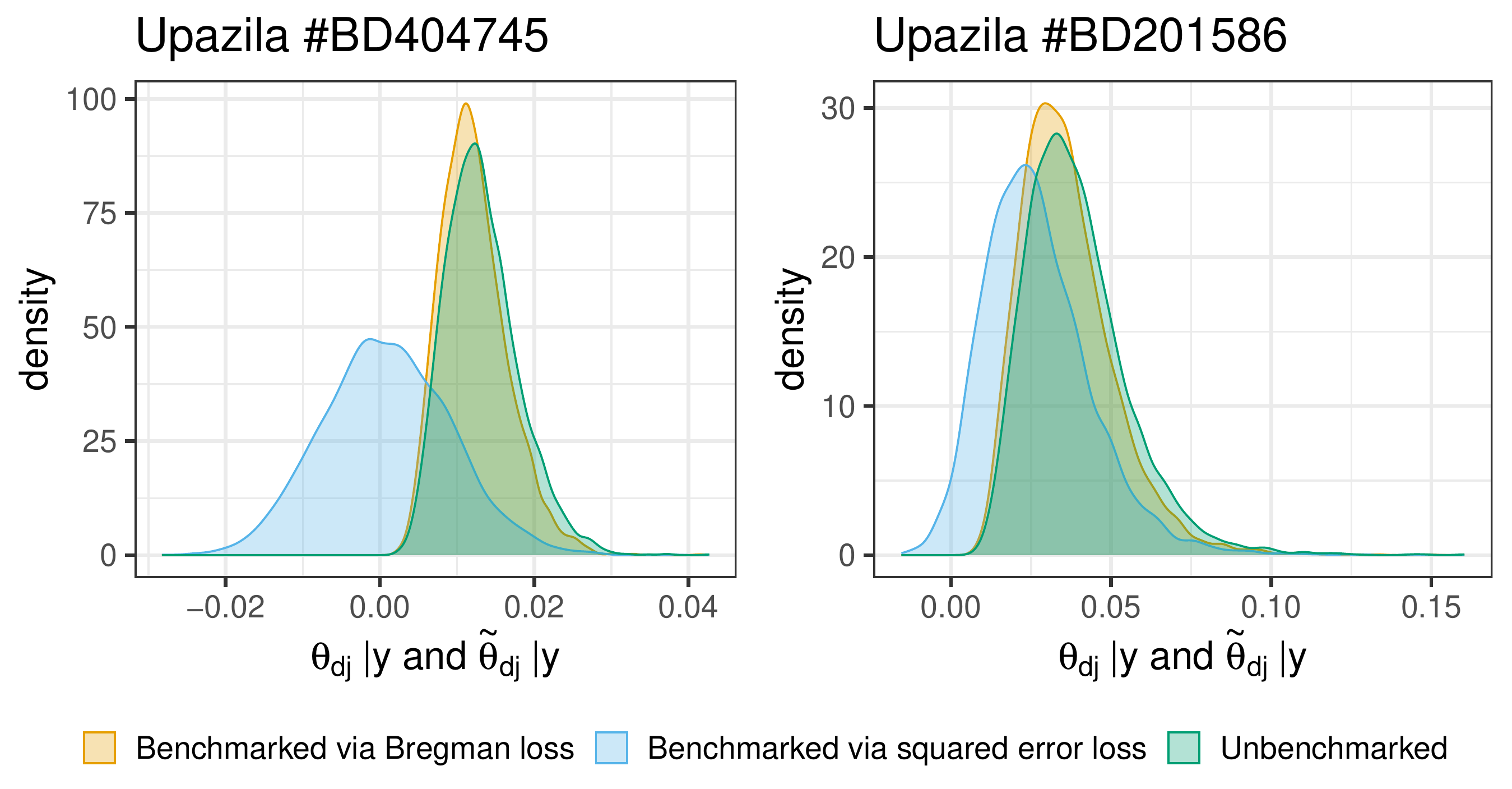}
    \caption{Posterior distributions of benchmarked and unbenchmarked estimates in two upazilas.}
    \label{fig:bench_sforo}
\end{figure}

The shrinking process induced by the model at both spatial levels is depicted in Figure \ref{fig:shrink}, by comparing survey estimates with model-based estimates. As expected, the shrinkage is stronger at the upazila level, given the downward precision of survey estimates. 
The differences between survey and model-based estimates may be high in case of very low effective sample sizes. Moreover, the right panel of Figure \ref{fig:shrink} shows that discrepancies at the area level may be also due to a high percentage of out-of-sample
sub-areas. In Figure \ref{fig:oos}, ensemble densities of sub-area estimates split up between in-sample and out-of-sample ones, together with a collection of the most relevant covariates. This illustrates that the composition of out-of-sample sub-areas is quite heterogeneous, incorporating both remote sub-areas (e.g. a large number of sub-areas on Chittagong Hill tracts) and urban or suburban ones. As a consequence, the out-of-sample poverty rates show a higher polarization between a set of urban and less poor sub-areas and a set of very poor rural ones. The latter is characterized by greater poverty levels than those fitted for in-sample sub-areas. This confirms that the sample selection process determining out-of-sample sub-areas can be affected by specific features.

\begin{figure}
    \centering
    \includegraphics[width = .8\linewidth]{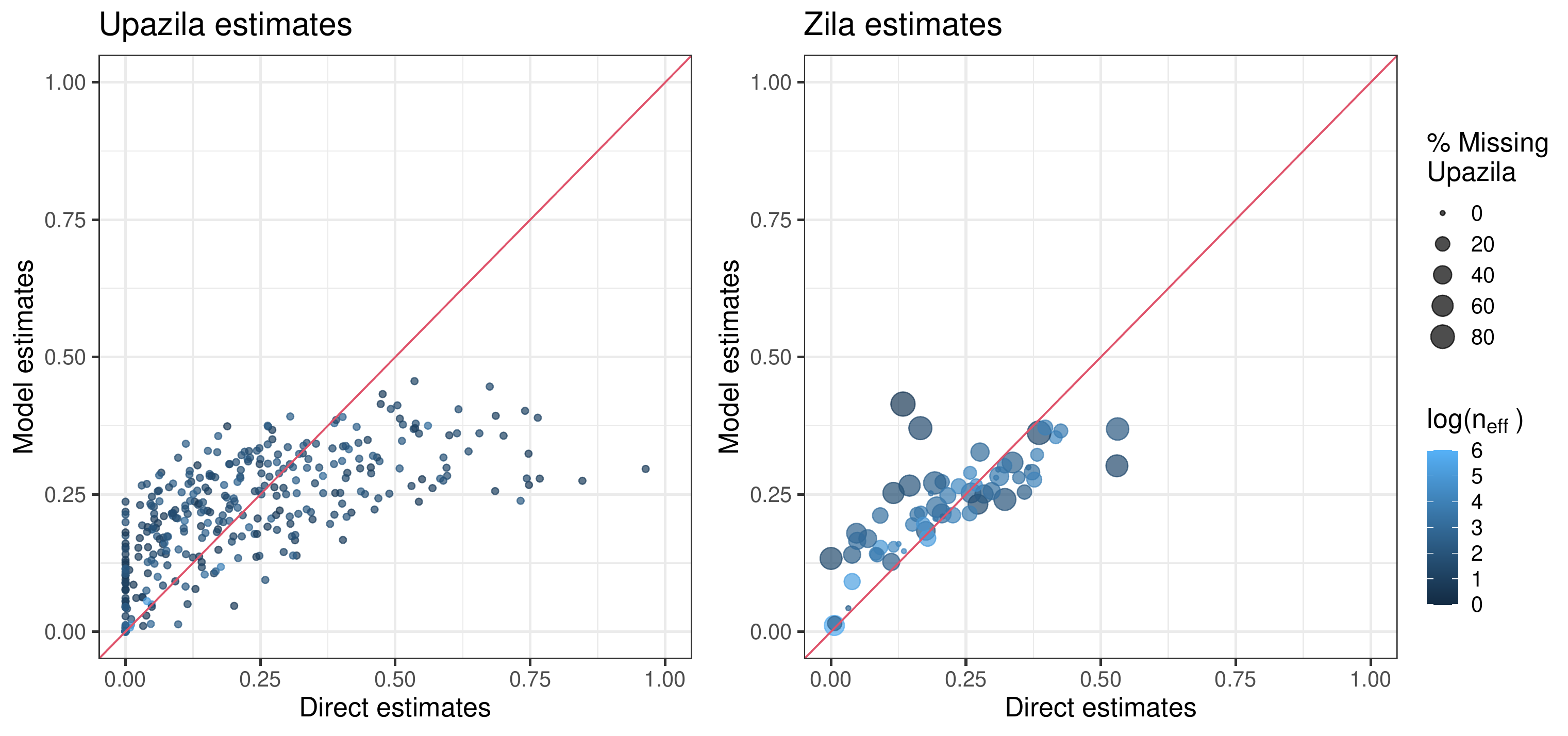}
    \caption{Shrinkage process at both spatial levels, bisector line in red. }
    \label{fig:shrink}
\end{figure}

The spatial distribution of model-based estimates is displayed in Figure \ref{fig:mapss}, enabling us to highlight the heterogeneity of poverty rates inside each zila. This is not clear by the zila map that averages out such estimates; whereas, at a finer level, the upazila map bears additional and valuable information.
The reduction of standard deviation of model-based estimates with respect to survey-based ones is of $51\%$ on average for zilas 
and $77\%$ for upazilas.

At a first glance, the poverty patterns are consistent with the literature \citep{kam2005spatial, imam2019small}, capturing consolidated dynamics. The metropolitan regions of Dhaka and Chittagong retain the lowest poverty levels, in contrast with peripheral areas. High poverty incidence domains overlap with ecologically poor zones \citep{kam2005spatial}. We can mention the territories more severely exposed to climate change effects and floodings such as the Haor depression (Sylhet basin lowlands) in the north-east, some areas at the edge of major rivers \citep{haque2015impact}, those exposed to droughts in the Rangpur division (north-west) and the Chittagong Hill Tract (south-east). On the other hand, the drought-prone southern Rajshahi and Khulna divisions experience lower poverty levels due to the good irrigation coverage \citep{kam2005spatial}. 
However, when disaggregating estimates at the upazila level, the picture becomes more clear. The poorest regions present a remarkable variety and it is possible to detect few upazilas having lower poverty levels than the neighbour ones. They correspond to main urban centers (indicated by red dots in Figure \ref{fig:mapss}) such as Rangpur, Dinajpur, and Saidpur cities in Rangpur division; Sylhet, Kishoreganj, and Mymensingh in the north-east side and Raozan in the south-east side.

\begin{figure}
    \centering
    \includegraphics[width = .8\linewidth]{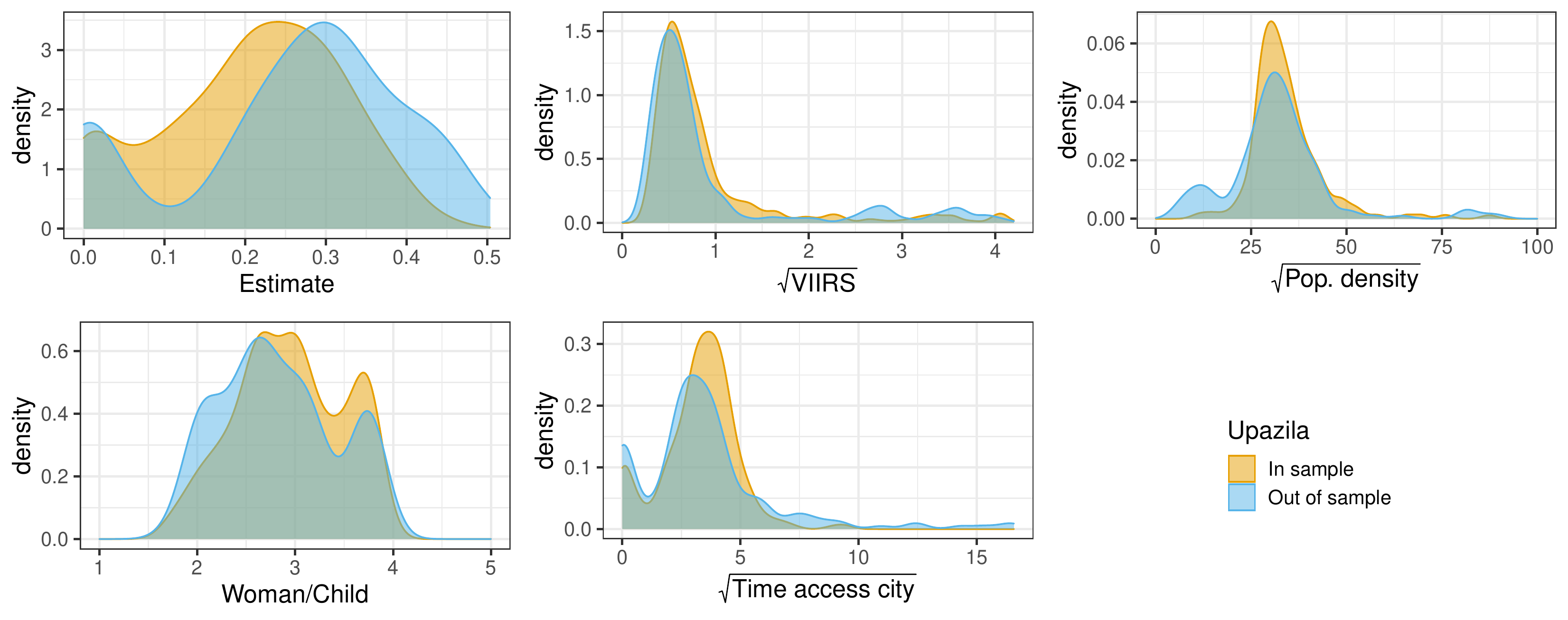}
    \caption{Densities of model-based estimates and the most relevant covariates compared between in-sample and out-of-sample sub-areas.}
    \label{fig:oos}
\end{figure}

\section{Conclusions}\label{sec:conclusions}
In this paper, we introduce a multi-scale modeling framework for poverty mapping at multiple spatial resolutions to avoid the scaling problem. The main aim is to produce reliable poverty maps that are coherent at different geographical layers. Such tools are valuable for poverty evaluation, poverty targeting and to develop place-based policies. On this line, we complement the proposal with a novel benchmarking algorithm that ensures the accordance of poverty estimates across levels. This allows to preserve poverty rate estimates and associated credible intervals within their support, restricted to $(0,1)$.


\begin{figure}
    \centering
    \includegraphics[width = .9\linewidth]{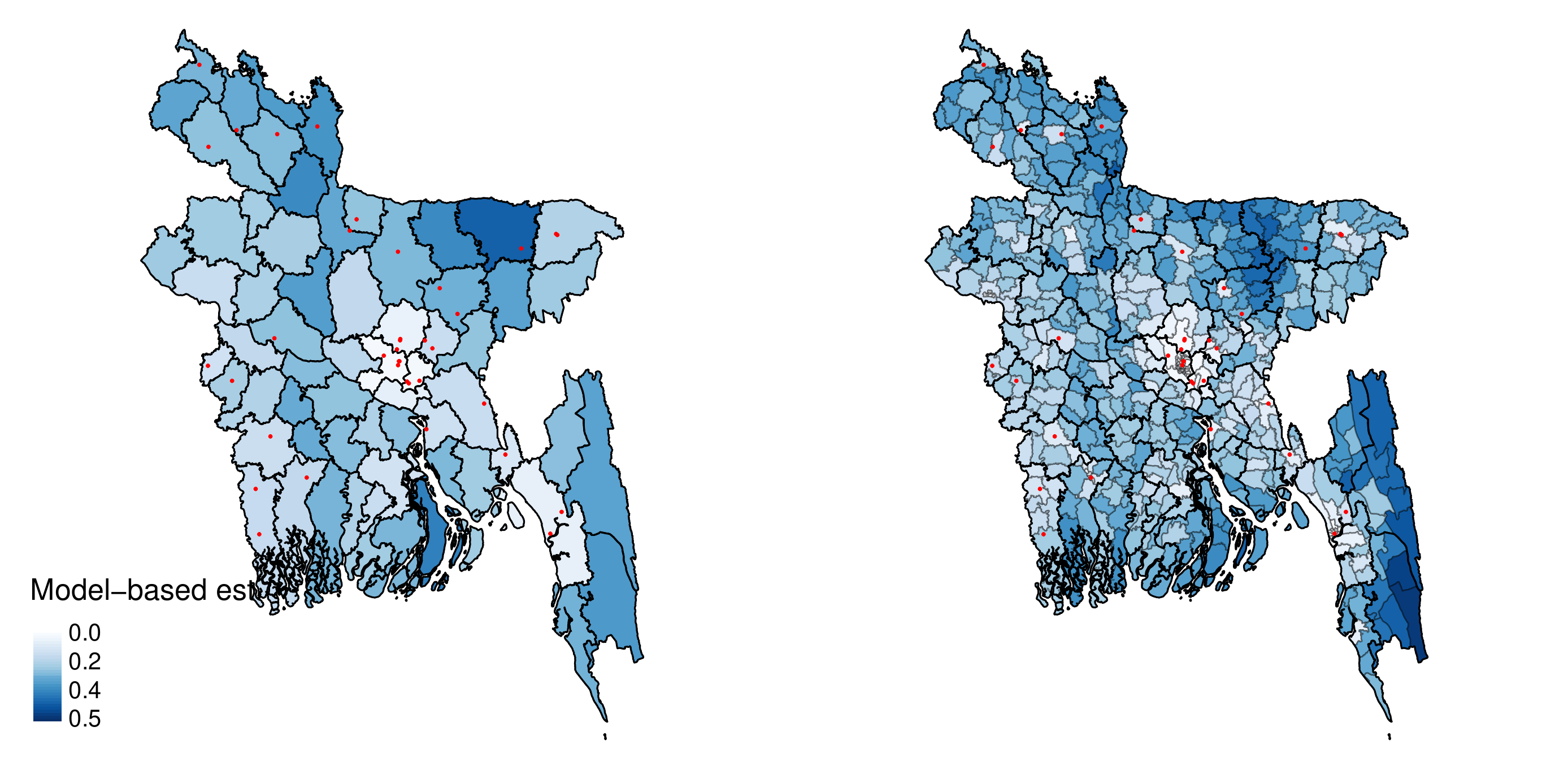}
    \caption{Model-based estimates at the zila level (left-hand side) and at the upazila level (right-hand side); red points indicate the 37 major cities.}
    \label{fig:mapss}
\end{figure}

The simulation study we carried out highlights the effectiveness of our proposal, showing better performances with respect to existing alternatives. In particular, this is evident for poverty rates laying in the tails of their distribution, which are less shrunken towards the average. In this way, our modeling strategy is able to retain domains with extremely high poverty rates: this has relevant practical implications since poverty hotspots need to be identified as being the target of poverty relief actions.
In our application on Bangladeshi data, rural and remote regions that mostly overlap with droughts- and floods-prone zones with high exposure to climate change effects, are classified as very poor. In this sense, geo-targeted policies may foster risk mitigation strategies (e.g., insurance) as well as  further expand irrigation coverages via specific water management policies \citep{pal2011evaluation, hossain2022smallholder}. Furthermore, our multi-scale poverty mapping approach permits us to spot critical zones located within larger administrative units. 

It is possible to extend our proposal to explicitly account for the spatial structure in the statistical model, similarly to \citet{aregay2016multiscale}. However, results of our application do not provide evidence of a residual spatial trend. This can be due to the spatial informative power of remote sensing covariates employed in the model.
Another future extension consists in considering also other quintiles of the WI index national distribution through a multivariate framework. Moreover, moving beyond poverty measurement, our proposal may be applied to map other well-being measures such as inequality indicators. 

\section*{Acknowledgments}
Work supported by the Data and Evidence to End Extreme Poverty (DEEP) research programme. DEEP is a consortium of the Universities of Cornell, Copenhagen, and Southampton led by Oxford Policy Management, in partnership with the World Bank -  Development Data Group and funded by the UK Foreign, Commonwealth \& Development Office. The work of Silvia De Nicolò was partially funded by the ALMA IDEA 2022 under Grant J45F21002000001 supported by the European Union - NextGenerationEU. The work of Aldo Gardini was partially supported by MUR on funds FSE REACT EU - PON R\&I 2014-2020 and PNR (D.M. 737/2021) for the RTDA\_GREEN project under Grant J41B21012140007.

\bibliographystyle{abbrvnat}
\bibliography{biblio}

\begin{thebibliography}{60}
\providecommand{\natexlab}[1]{#1}
\providecommand{\url}[1]{\texttt{#1}}
\expandafter\ifx\csname urlstyle\endcsname\relax
  \providecommand{\doi}[1]{doi: #1}\else
  \providecommand{\doi}{doi: \begingroup \urlstyle{rm}\Url}\fi

\bibitem[Aitchison(1992)]{aitchison1992criteria}
J.~Aitchison.
\newblock On criteria for measures of compositional difference.
\newblock \emph{Mathematical Geology}, 24\penalty0 (4):\penalty0 365--379,
  1992.

\bibitem[Allard and Allard(2017)]{allard2017places}
S.~W. Allard and S.~Allard.
\newblock \emph{Places in need: The changing geography of poverty}.
\newblock Russell Sage Foundation, 2017.

\bibitem[Aregay et~al.(2016)Aregay, Lawson, Faes, Kirby, Carroll, and
  Watjou]{aregay2016multiscale}
M.~Aregay, A.~B. Lawson, C.~Faes, R.~S. Kirby, R.~Carroll, and K.~Watjou.
\newblock Multiscale measurement error models for aggregated small area health
  data.
\newblock \emph{Statistical methods in medical research}, 25\penalty0
  (4):\penalty0 1201--1223, 2016.

\bibitem[Aregay et~al.(2017)Aregay, Lawson, Faes, Kirby, Carroll, and
  Watjou]{aregay2017comparing}
M.~Aregay, A.~B. Lawson, C.~Faes, R.~S. Kirby, R.~Carroll, and K.~Watjou.
\newblock Comparing multilevel and multiscale convolution models for small area
  aggregated health data.
\newblock \emph{Spatial and Spatio-temporal Epidemiology}, 22:\penalty0 39--49,
  2017.

\bibitem[Banerjee et~al.(2005)Banerjee, Guo, and Wang]{banerjee2005optimality}
A.~Banerjee, X.~Guo, and H.~Wang.
\newblock On the optimality of conditional expectation as a bregman predictor.
\newblock \emph{IEEE Transactions on Information Theory}, 51\penalty0
  (7):\penalty0 2664--2669, 2005.

\bibitem[Bell et~al.(2013)Bell, Datta, and Ghosh]{bell2013benchmarking}
W.~R. Bell, G.~S. Datta, and M.~Ghosh.
\newblock Benchmarking small area estimators.
\newblock \emph{Biometrika}, 100\penalty0 (1):\penalty0 189--202, 2013.

\bibitem[Benedetti et~al.(2022)Benedetti, Berrocal, and
  Little]{benedetti2022accounting}
M.~H. Benedetti, V.~J. Berrocal, and R.~J. Little.
\newblock Accounting for survey design in bayesian disaggregation of
  survey-based areal estimates of proportions: An application to the american
  community survey.
\newblock \emph{The Annals of Applied Statistics}, 16\penalty0 (4):\penalty0
  2201--2230, 2022.

\bibitem[Bigman and Fofack(2000)]{bigman2000geographical}
D.~Bigman and H.~Fofack.
\newblock Geographical targeting for poverty alleviation: An introduction to
  the special issue.
\newblock \emph{The World Bank Economic Review}, 14\penalty0 (1):\penalty0
  129--145, 2000.

\bibitem[Bikauskaite et~al.(2022)Bikauskaite, Molina, and
  Morales]{bikauskaite2022multivariate}
A.~Bikauskaite, I.~Molina, and D.~Morales.
\newblock Multivariate mixture model for small area estimation of poverty
  indicators.
\newblock \emph{Journal of the Royal Statistical Society Series A: Statistics
  in Society}, 185\penalty0 (Supplement\_2):\penalty0 S724--S755, 2022.

\bibitem[Carpenter et~al.(2017)Carpenter, Gelman, Hoffman, Lee, Goodrich,
  Betancourt, Brubaker, Guo, Li, and Riddell]{carpenter2017stan}
B.~Carpenter, A.~Gelman, M.~D. Hoffman, D.~Lee, B.~Goodrich, M.~Betancourt,
  M.~Brubaker, J.~Guo, P.~Li, and A.~Riddell.
\newblock Stan: A probabilistic programming language.
\newblock \emph{Journal of statistical software}, 76\penalty0 (1), 2017.

\bibitem[Casas-Cordero~Valencia et~al.(2016)Casas-Cordero~Valencia, Encina, and
  Lahiri]{casas2016poverty}
C.~Casas-Cordero~Valencia, J.~Encina, and P.~Lahiri.
\newblock Poverty mapping for the chilean comunas.
\newblock \emph{Analysis of Poverty Data by Small Area Estimation}, pages
  379--404, 2016.

\bibitem[Christiaensen and Todo(2014)]{christiaensen2014poverty}
L.~Christiaensen and Y.~Todo.
\newblock Poverty reduction during the rural--urban transformation--the role of
  the missing middle.
\newblock \emph{World Development}, 63:\penalty0 43--58, 2014.

\bibitem[Corral et~al.(2022)Corral, Molina, Cojocaru, and
  Segovia]{corral2022guidelines}
P.~Corral, I.~Molina, A.~Cojocaru, and S.~Segovia.
\newblock Guidelines to small area estimation for poverty mapping.
\newblock 2022.

\bibitem[Datta et~al.(2011)Datta, Ghosh, Steorts, and
  Maples]{datta2011bayesian}
G.~Datta, M.~Ghosh, R.~Steorts, and J.~Maples.
\newblock Bayesian benchmarking with applications to small area estimation.
\newblock \emph{Test}, 20\penalty0 (3):\penalty0 574--588, 2011.

\bibitem[De~Nicol{\`o} et~al.(2022)De~Nicol{\`o}, Fabrizi, and
  Gardini]{de2022extended}
S.~De~Nicol{\`o}, E.~Fabrizi, and A.~Gardini.
\newblock Extended beta models for poverty mapping. an application integrating
  survey and remote sensing data in {B}angladesh.
\newblock In \emph{Quaderni di Dipartimento. Serie Ricerche}. 2022.

\bibitem[Dunson and Neelon(2003)]{dunson2003bayesian}
D.~B. Dunson and B.~Neelon.
\newblock Bayesian inference on order-constrained parameters in generalized
  linear models.
\newblock \emph{Biometrics}, 59\penalty0 (2):\penalty0 286--295, 2003.

\bibitem[Elbers et~al.(2007)Elbers, Fujii, Lanjouw, {\"O}zler, and
  Yin]{elbers2007poverty}
C.~Elbers, T.~Fujii, P.~Lanjouw, B.~{\"O}zler, and W.~Yin.
\newblock Poverty alleviation through geographic targeting: How much does
  disaggregation help?
\newblock \emph{Journal of Development Economics}, 83\penalty0 (1):\penalty0
  198--213, 2007.

\bibitem[Erciulescu et~al.(2019)Erciulescu, Cruze, and
  Nandram]{erciulescu2019model}
A.~L. Erciulescu, N.~B. Cruze, and B.~Nandram.
\newblock Model-based county level crop estimates incorporating auxiliary
  sources of information.
\newblock \emph{Journal of the Royal Statistical Society: Series A (Statistics
  in Society)}, \penalty0 (1):\penalty0 283--303, 2019.

\bibitem[Fabic et~al.(2012)Fabic, Choi, and Bird]{fabic2012systematic}
M.~S. Fabic, Y.~Choi, and S.~Bird.
\newblock A systematic review of demographic and health surveys: data
  availability and utilization for research.
\newblock \emph{Bulletin of the World Health Organization}, 90:\penalty0
  604--612, 2012.

\bibitem[Fabrizi et~al.(2016{\natexlab{a}})Fabrizi, Ferrante, and
  Trivisano]{fabrizi2016hierarchical}
E.~Fabrizi, M.~Ferrante, and C.~Trivisano.
\newblock Hierarchical beta regression models for the estimation of poverty and
  inequality parameters in small areas.
\newblock \emph{Analysis of Poverty Data by Small Area Methods. John Wiley and
  Sons}, pages 299--314, 2016{\natexlab{a}}.

\bibitem[Fabrizi et~al.(2016{\natexlab{b}})Fabrizi, Ferrante, and
  Trivisano]{fabrizi2016bayesian}
E.~Fabrizi, M.~R. Ferrante, and C.~Trivisano.
\newblock Bayesian beta regression models for the estimation of poverty and
  inequality parameters in small areas.
\newblock \emph{Analysis of poverty data by small area estimation}, pages
  299--314, 2016{\natexlab{b}}.

\bibitem[Fan and Cho(2021)]{fan2021paths}
S.-g. Fan and E.~E. Cho.
\newblock Paths out of poverty: International experience.
\newblock \emph{Journal of Integrative Agriculture}, 20\penalty0 (4):\penalty0
  857--867, 2021.

\bibitem[Ferrari and Cribari-Neto(2004)]{ferrari2004beta}
S.~Ferrari and F.~Cribari-Neto.
\newblock Beta regression for modelling rates and proportions.
\newblock \emph{Journal of Applied Statistics}, 31\penalty0 (7):\penalty0
  799--815, 2004.

\bibitem[Galasso and Ravallion(2005)]{galasso2005decentralized}
E.~Galasso and M.~Ravallion.
\newblock Decentralized targeting of an antipoverty program.
\newblock \emph{Journal of Public economics}, 89\penalty0 (4):\penalty0
  705--727, 2005.

\bibitem[Gauci(2005)]{gauci2005spatial}
A.~Gauci.
\newblock Spatial maps. targeting \& mapping poverty.
\newblock \emph{London: United Nations. Economic Commission for Africa}, 2005.

\bibitem[Ghosh et~al.(2015)Ghosh, Kubokawa, and Kawakubo]{ghosh2015benchmarked}
M.~Ghosh, T.~Kubokawa, and Y.~Kawakubo.
\newblock Benchmarked empirical bayes methods in multiplicative area-level
  models with risk evaluation.
\newblock \emph{Biometrika}, 102\penalty0 (3):\penalty0 647--659, 2015.

\bibitem[H{\'a}jek(1971)]{hajek1971discussion}
J.~H{\'a}jek.
\newblock Discussion of ‘an essay on the logical foundations of survey
  sampling, part i’, by d. basu.
\newblock \emph{Foundations of Statistical Inference}, page 326, 1971.

\bibitem[Hall et~al.(2023)Hall, Dompae, Wahab, and Dzanku]{hall2023review}
O.~Hall, F.~Dompae, I.~Wahab, and F.~M. Dzanku.
\newblock A review of machine learning and satellite imagery for poverty
  prediction: Implications for development research and applications.
\newblock \emph{Journal of International Development}, 2023.

\bibitem[Haque and Jahan(2015)]{haque2015impact}
A.~Haque and S.~Jahan.
\newblock Impact of flood disasters in bangladesh: A multi-sector regional
  analysis.
\newblock \emph{International Journal of Disaster Risk Reduction}, 13:\penalty0
  266--275, 2015.

\bibitem[Hossain et~al.(2022)Hossain, Alam, Fahad, Sarker, Moniruzzaman, and
  Rabbany]{hossain2022smallholder}
M.~S. Hossain, G.~M. Alam, S.~Fahad, T.~Sarker, M.~Moniruzzaman, and M.~G.
  Rabbany.
\newblock Smallholder farmers’ willingness to pay for flood insurance as
  climate change adaptation strategy in northern bangladesh.
\newblock \emph{Journal of Cleaner Production}, 338:\penalty0 130584, 2022.

\bibitem[Imam et~al.(2019)Imam, Islam, Alam, Hossain, and Das]{imam2019small}
M.~F. Imam, M.~A. Islam, M.~A. Alam, M.~J. Hossain, and S.~Das.
\newblock Small area estimation of poverty in rural bangladesh.
\newblock \emph{The Bangladesh Journal of Agricultural Economics}, 40\penalty0
  (1\&2):\penalty0 1--16, 2019.

\bibitem[Janicki(2020)]{janicki2020properties}
R.~Janicki.
\newblock Properties of the beta regression model for small area estimation of
  proportions and application to estimation of poverty rates.
\newblock \emph{Communications in Statistics-Theory and Methods}, 49\penalty0
  (9):\penalty0 2264--2284, 2020.

\bibitem[Janicki and Vesper(2017)]{janicki2017benchmarking}
R.~Janicki and A.~Vesper.
\newblock Benchmarking techniques for reconciling bayesian small area models at
  distinct geographic levels.
\newblock \emph{Statistical Methods \& Applications}, 26\penalty0 (4):\penalty0
  557--581, 2017.

\bibitem[Jean et~al.(2016)Jean, Burke, Xie, Davis, Lobell, and
  Ermon]{jean2016combining}
N.~Jean, M.~Burke, M.~Xie, W.~M. Davis, D.~B. Lobell, and S.~Ermon.
\newblock Combining satellite imagery and machine learning to predict poverty.
\newblock \emph{Science}, 353\penalty0 (6301):\penalty0 790--794, 2016.

\bibitem[Kam et~al.(2005)Kam, Hossain, Bose, and Villano]{kam2005spatial}
S.-P. Kam, M.~Hossain, M.~L. Bose, and L.~S. Villano.
\newblock Spatial patterns of rural poverty and their relationship with
  welfare-influencing factors in bangladesh.
\newblock \emph{Food Policy}, 30\penalty0 (5-6):\penalty0 551--567, 2005.

\bibitem[Kolaczyk and Huang(2001)]{kolaczyk2001multiscale}
E.~D. Kolaczyk and H.~Huang.
\newblock Multiscale statistical models for hierarchical spatial aggregation.
\newblock \emph{Geographical Analysis}, 33\penalty0 (2):\penalty0 95--118,
  2001.

\bibitem[Kraay and McKenzie(2014)]{kraay2014poverty}
A.~Kraay and D.~McKenzie.
\newblock Do poverty traps exist? assessing the evidence.
\newblock \emph{Journal of Economic Perspectives}, 28\penalty0 (3):\penalty0
  127--148, 2014.

\bibitem[Liu et~al.(2014)Liu, Lahiri, and Kalton]{liu2014hierarchical}
B.~Liu, P.~Lahiri, and G.~Kalton.
\newblock Hierarchical bayes modeling of survey-weighted small area
  proportions.
\newblock \emph{Survey Methodology}, 40\penalty0 (1):\penalty0 1--14, 2014.

\bibitem[Liu et~al.(2017)Liu, Liu, and Zhou]{liu2017spatio}
Y.~Liu, J.~Liu, and Y.~Zhou.
\newblock Spatio-temporal patterns of rural poverty in china and targeted
  poverty alleviation strategies.
\newblock \emph{Journal of rural studies}, 52:\penalty0 66--75, 2017.

\bibitem[Louie and Kolaczyk(2006)]{louie2006multiscale}
M.~M. Louie and E.~D. Kolaczyk.
\newblock A multiscale method for disease mapping in spatial epidemiology.
\newblock \emph{Statistics in medicine}, 25\penalty0 (8):\penalty0 1287--1306,
  2006.

\bibitem[Molina et~al.(2014)Molina, Nandram, and Rao]{molina2014small}
I.~Molina, B.~Nandram, and J.~Rao.
\newblock Small area estimation of general parameters with application to
  poverty indicators: A hierarchical bayes approach.
\newblock \emph{The Annals of Applied Statistics}, 8\penalty0 (2):\penalty0
  852--885, 2014.

\bibitem[Okonek and Wakefield(2022)]{okonek2022computationally}
T.~Okonek and J.~Wakefield.
\newblock A computationally efficient approach to fully bayesian benchmarking.
\newblock \emph{arXiv preprint arXiv:2203.12195}, 2022.

\bibitem[Pal et~al.(2011)Pal, Adeloye, Babel, and Das~Gupta]{pal2011evaluation}
S.~K. Pal, A.~J. Adeloye, M.~S. Babel, and A.~Das~Gupta.
\newblock Evaluation of the effectiveness of water management policies in
  bangladesh.
\newblock \emph{Water Resources Development}, 27\penalty0 (02):\penalty0
  401--417, 2011.

\bibitem[Patra(2019)]{patra2019constrained}
S.~Patra.
\newblock \emph{Constrained Bayesian inference through posterior projection
  with applications}.
\newblock PhD thesis, 2019.

\bibitem[Piironen and Vehtari(2017)]{horseshoe}
J.~Piironen and A.~Vehtari.
\newblock {Sparsity information and regularization in the horseshoe and other
  shrinkage priors}.
\newblock \emph{Electronic Journal of Statistics}, 11\penalty0 (2):\penalty0
  5018 -- 5051, 2017.
\newblock \doi{10.1214/17-EJS1337SI}.

\bibitem[Pirani(2014)]{pirani2014wealth}
E.~Pirani.
\newblock Wealth index.
\newblock \emph{Encyclopedia of Quality of Life and Well-Being Research.
  Springer, Dordrecht. https://doi. org/10.1007/978-94-007-0753-5\_3202}, 2014.

\bibitem[Poirier et~al.(2020)Poirier, Gr{\'e}pin, and
  Grignon]{poirier2020approaches}
M.~J. Poirier, K.~A. Gr{\'e}pin, and M.~Grignon.
\newblock Approaches and alternatives to the wealth index to measure
  socioeconomic status using survey data: a critical interpretive synthesis.
\newblock \emph{Social Indicators Research}, 148\penalty0 (1):\penalty0 1--46,
  2020.

\bibitem[Pratesi and Salvati(2016)]{pratesi2016introduction}
M.~Pratesi and N.~Salvati.
\newblock Introduction on measuring poverty at local level using small area
  estimation methods.
\newblock In \emph{Analysis of poverty data by small area estimation},
  chapter~1, pages 1--18. Wiley Online Library, 2016.

\bibitem[Puurbalanta(2021)]{puurbalanta2021clipped}
R.~Puurbalanta.
\newblock A clipped gaussian geo-classification model for poverty mapping.
\newblock \emph{Journal of Applied Statistics}, 48\penalty0 (10):\penalty0
  1882--1895, 2021.

\bibitem[Rao and Molina(2015)]{rao2015small}
J.~N. Rao and I.~Molina.
\newblock \emph{Small area estimation}.
\newblock John Wiley \& Sons, 2015.

\bibitem[Schmid et~al.(2017)Schmid, Bruckschen, Salvati, and
  Zbiranski]{schmid2017constructing}
T.~Schmid, F.~Bruckschen, N.~Salvati, and T.~Zbiranski.
\newblock Constructing sociodemographic indicators for national statistical
  institutes by using mobile phone data: estimating literacy rates in senegal.
\newblock \emph{Journal of the Royal Statistical Society: Series A (Statistics
  in Society)}, 180\penalty0 (4):\penalty0 1163--1190, 2017.

\bibitem[Sen et~al.(2018)Sen, Patra, and Dunson]{sen2018constrained}
D.~Sen, S.~Patra, and D.~Dunson.
\newblock Constrained inference through posterior projections.
\newblock \emph{arXiv preprint arXiv:1812.05741}, 2018.

\bibitem[Sohnesen et~al.(2022)Sohnesen, Fisker, and
  Malmgren-Hansen]{sohnesen2022using}
T.~P. Sohnesen, P.~Fisker, and D.~Malmgren-Hansen.
\newblock Using satellite data to guide urban poverty reduction.
\newblock \emph{Review of Income and Wealth}, 68:\penalty0 S282--S294, 2022.

\bibitem[Steele et~al.(2017)Steele, Sunds{\o}y, Pezzulo, Alegana, Bird,
  Blumenstock, Bjelland, Eng{\o}-Monsen, De~Montjoye, Iqbal,
  et~al.]{steele2017mapping}
J.~E. Steele, P.~R. Sunds{\o}y, C.~Pezzulo, V.~A. Alegana, T.~J. Bird,
  J.~Blumenstock, J.~Bjelland, K.~Eng{\o}-Monsen, Y.-A. De~Montjoye, A.~M.
  Iqbal, et~al.
\newblock Mapping poverty using mobile phone and satellite data.
\newblock \emph{Journal of The Royal Society Interface}, 14\penalty0
  (127):\penalty0 20160690, 2017.

\bibitem[Torabi and Rao(2014)]{torabi2014small}
M.~Torabi and J.~Rao.
\newblock On small area estimation under a sub-area level model.
\newblock \emph{Journal of Multivariate Analysis}, 127:\penalty0 36--55, 2014.

\bibitem[Tzavidis et~al.(2018)Tzavidis, Zhang, Luna, Schmid, Rojas-Perilla,
  Gordon, Williamson, King, Ranalli, Smith, et~al.]{tzavidis2018start}
N.~Tzavidis, L.-C. Zhang, A.~Luna, T.~Schmid, N.~Rojas-Perilla, l.~R. Gordon,
  P.~Williamson, T.~King, M.~G. Ranalli, P.~A. Smith, et~al.
\newblock From start to finish.
\newblock \emph{Journal of the Royal Statistical Society. Series A (Statistics
  in Society)}, 181\penalty0 (4):\penalty0 927--979, 2018.

\bibitem[Vehtari et~al.(2017)Vehtari, Gelman, and Gabry]{vehtari2017practical}
A.~Vehtari, A.~Gelman, and J.~Gabry.
\newblock Practical bayesian model evaluation using leave-one-out
  cross-validation and waic.
\newblock \emph{Statistics and Computing}, 27\penalty0 (5):\penalty0
  1413--1432, 2017.

\bibitem[Zhang and Bryant(2020)]{zhang2020fully}
J.~L. Zhang and J.~Bryant.
\newblock Fully bayesian benchmarking of small area estimation models.
\newblock \emph{Journal of official statistics}, 36\penalty0 (1):\penalty0
  197--223, 2020.

\bibitem[Zhao et~al.(2019)Zhao, Yu, Liu, Chen, Li, Wang, and
  Wu]{zhao2019estimation}
X.~Zhao, B.~Yu, Y.~Liu, Z.~Chen, Q.~Li, C.~Wang, and J.~Wu.
\newblock Estimation of poverty using random forest regression with
  multi-source data: A case study in bangladesh.
\newblock \emph{Remote Sensing}, 11\penalty0 (4):\penalty0 375, 2019.

\bibitem[Zhou and Liu(2022)]{zhou2022geography}
Y.~Zhou and Y.~Liu.
\newblock The geography of poverty: Review and research prospects.
\newblock \emph{Journal of Rural Studies}, 93:\penalty0 408--416, 2022.

\end{thebibliography}

\end{document}